\RequirePackage{fix-cm}
\documentclass[twocolumn]{svjour3}          
\smartqed  
\usepackage{graphicx}
\usepackage{mathptmx}      
\usepackage[T1]{fontenc}
\usepackage[utf8]{inputenc}
\usepackage[english]{babel}
\usepackage{csquotes}
\usepackage[
	disable
]{todonotes}
\usepackage[noend]{algpseudocode}
\usepackage[Algorithm]{algorithm}
\usepackage{booktabs}
\usepackage{tikz}
	\usetikzlibrary{positioning,arrows,matrix}
\usepackage{amssymb}
\usepackage{mathtools}
\usepackage{hyperref}
\hypersetup{hidelinks,
  colorlinks=true,
  allcolors=blue,
  pdfstartview=Fit,
  breaklinks=true}
\usepackage[all]{hypcap}

\makeatletter
\newcommand*{\da@rightarrow}{\mathchar"0\hexnumber@\symAMSa 4B }
\newcommand*{\da@leftarrow}{\mathchar"0\hexnumber@\symAMSa 4C }
\newcommand*{\xdashrightarrow}[2][]{%
  \mathrel{%
    \mathpalette{\da@xarrow{#1}{#2}{}\da@rightarrow{\,}{}}{}%
  }%
}
\newcommand{\xdashleftarrow}[2][]{%
  \mathrel{%
    \mathpalette{\da@xarrow{#1}{#2}\da@leftarrow{}{}{\,}}{}%
  }%
}
\newcommand*{\da@xarrow}[7]{%
  \sbox0{$\ifx#7\scriptstyle\scriptscriptstyle\else\scriptstyle\fi#5#1#6\m@th$}%
  \sbox2{$\ifx#7\scriptstyle\scriptscriptstyle\else\scriptstyle\fi#5#2#6\m@th$}%
  \sbox4{$#7\dabar@\m@th$}%
  \dimen@=\wd0 %
  \ifdim\wd2 >\dimen@
    \dimen@=\wd2 %
  \fi
  \count@=2 %
  \def\da@bars{\dabar@\dabar@}%
  \@whiledim\count@\wd4<\dimen@\do{%
    \advance\count@\@ne
    \expandafter\def\expandafter\da@bars\expandafter{%
      \da@bars
      \dabar@ 
    }%
  }%
  \mathrel{#3}%
  \mathrel{%
    \mathop{\da@bars}\limits
    \ifx\\#1\\%
    \else
      _{\copy0}%
    \fi
    \ifx\\#2\\%
    \else
      ^{\copy2}%
    \fi
  }%
  \mathrel{#4}%
}
\makeatother

\newcommand{\Lcap}{L_{\cap}}
\newcommand{\Rcap}{R_{\cap}}

\spnewtheorem{fact}[theorem]{Fact}{\bfseries}{\itshape}
\spnewtheorem{construction}[definition]{Construction}{\bf}{\rm}
\begin{document}

\title{Avoiding Unnecessary Information Loss}
\subtitle{Correct and Efficient Model Synchronization Based on Triple Graph Grammars}

\author{%
	Lars Fritsche \and 
	Jens Kosiol \and 
	Andy Schürr \and 
	Gabriele Taentzer 
}
\authorrunning{L. Fritsche et al.}
\institute{L. Fritsche \at
						Real-Time Systems Lab, TU Darmstadt, Darmstadt, Germany\\
						\email{lars.fritsche@es.tu-darmstadt.de}
					\and
					J. Kosiol \at
						Faculty of Mathematics and Computer Science, Philipps-Universität Marburg, Marburg, Germany\\
						\email{kosiolje@mathematik.uni-marburg.de}
					\and 
					A. Schürr \at
						Real-Time Systems Lab, TU Darmstadt, Darmstadt, Germany\\
						\email{andy.schuerr@es.tu-darmstadt.de}
					\and
					G. Taentzer \at
						Faculty of Mathematics and Computer Science, Philipps-Universität Marburg, Marburg, Germany\\
						\email{taentzer@mathematik.uni-marburg.de}
}


\maketitle

\newcommand{\myfigureshrinker}{\vspace{0cm}}

\newcommand{\marked}{$\text{\rlap{$\checkmark$}}\square$}
\newcommand{\marking}{$\square \rightarrow $ \marked}

\newcommand{\firstTGGRule}{\textit{Root-Rule}}
\newcommand{\secondTGGRule}{\textit{Sub-Rule}}
\newcommand{\thirdTGGRule}{\textit{Leaf-Rule}}

\newcommand{\firstTGGForwardRule}{\textit{Root-FWD-Rule}}
\newcommand{\secondTGGForwardRule}{\textit{Sub-FWD-Rule}}
\newcommand{\thirdTGGForwardRule}{\textit{Leaf-FWD-Rule}}

\newcommand{\shortcutRule}{short-cut rule}
\newcommand{\shortcutRules}{short-cut rules}
\newcommand{\ShortcutRules}{Short-cut rules}
\newcommand{\ShortcutRule}{Short-cut rule}
\newcommand{\operationalSCRule}{operational short-cut rule}
\newcommand{\operationalSCRules}{operational short-cut rules}
\newcommand{\forwardShortcutRules}{forward short-cut rules}

\newcommand{\firstSCRule}{\textit{Connect-Root-SC-Rule}}
\newcommand{\secondSCRule}{\textit{Make-Root-SC-Rule}}
\newcommand{\thirdSCRule}{\textit{Move-To-New-Sub-SC-Rule}}
\newcommand{\fourthSCRule}{\textit{Delete-Middle-SC-Rule}}

\newcommand{\RepairRule}{\textit{Repair rule}}
\newcommand{\RepairRules}{\textit{Repair rules}}
\newcommand{\repairRules}{\textit{repair rules}}
\newcommand{\repairRule}{\textit{repair rule}}

\newcommand{\firstSCSourceRule}{\textit{Connect-Root-Source-Rule}}
\newcommand{\secondSCSourceRule}{\textit{Make-Root-Source-Rule}}
\newcommand{\thirdSCSouceRule}{\textit{Move-To-New-Sub-Source-Rule}}
\newcommand{\fourthSCSourceRule}{\textit{Delete-Middle-Source-Rule}}

\newcommand{\firstRRule}{\textit{Connect-Root-Repair-Rule}}
\newcommand{\secondRRule}{\textit{Make-Root-Repair-Rule}}
\newcommand{\thirdRRule}{\textit{Move-To-New-Sub-Repair-Rule}}
\newcommand{\fourthRRule}{\textit{Delete-Middle-Repair-Rule}}

\newcommand{\docModel}{\textit{DocModel}}
\newcommand{\docModels}{\textit{DocModels}}
\newcommand{\folder}{\textit{Folder}}
\newcommand{\folders}{\textit{Folders}}
\newcommand{\package}{\textit{Package}}
\newcommand{\packages}{\textit{Packages}}
\newcommand{\doc}{\textit{Doc-File}}
\newcommand{\docs}{\textit{Doc-Files}}
\newcommand{\methodEntry}{\textit{MethodEntry}}
\newcommand{\methodEntries}{\textit{MethodEntries}}
\newcommand{\parameter}{\textit{Parameter}}
\newcommand{\parameters}{\textit{Parameters}}

\newcommand{\model}{\textit{Model}}
\newcommand{\class}{\textit{Class}}
\newcommand{\classes}{\textit{Classes}}
\newcommand{\classDec}{\textit{ClassDeclaration}}
\newcommand{\classDecs}{\textit{ClassesDeclarations}}
\newcommand{\interface}{\textit{Interface}}
\newcommand{\interfaces}{\textit{Interfaces}}
\newcommand{\interfaceDec}{\textit{InterfaceDeclaration}}
\newcommand{\interfaceDecs}{\textit{InterfacesDeclaration}}
\newcommand{\enum}{\textit{Enum}}
\newcommand{\enums}{\textit{Enums}}
\newcommand{\enumDec}{\textit{EnumDeclaration}}
\newcommand{\enumDecs}{\textit{EnumDeclaration}}
\newcommand{\field}{\textit{Field}}
\newcommand{\fields}{\textit{Fields}}
\newcommand{\method}{\textit{Method}}
\newcommand{\methods}{\textit{Methods}}
\newcommand{\methodDec}{\textit{MethodDeclaration}}
\newcommand{\methodDecs}{\textit{MethodDeclarations}}
\newcommand{\typeAccess}{\textit{TypeAccess}}
\newcommand{\typeAccesses}{\textit{TypeAccesses}}

\newcommand{\rootP}{\texttt{rootP}}
\newcommand{\f}{\texttt{f}}
\newcommand{\p}{\texttt{p}}
\newcommand{\subP}{\texttt{subP}}
\newcommand{\subPDoc}{\texttt{subPDoc}}
\newcommand{\leafP}{\texttt{leafP}}
\newcommand{\cl}{\texttt{c}}
\newcommand{\ssp}{\texttt{sp}}
\newcommand{\op}{\texttt{op}}
\newcommand{\of}{\texttt{of}}
\newcommand{\np}{\texttt{np}}
\newcommand{\nf}{\texttt{nf}}

\newcommand{\rootF}{\texttt{rootF}}
\newcommand{\subF}{\texttt{subF}}
\newcommand{\leafF}{\texttt{leafF}}
\newcommand{\cDoc}{\texttt{cDoc}}
\newcommand{\leafPDoc}{\texttt{leafPDoc}}
\newcommand{\ssf}{\texttt{sf}}
\newcommand{\dd}{\texttt{d}}

\begin{abstract}
Model synchronization, i.e., the task of restoring consistency between two interrelated models after a model change, is a challenging task.
Triple Graph Grammars (TGGs) specify model consistency by means of rules that describe how to create consistent pairs of models. 
These rules can be used to automatically derive further rules, which describe how to propagate changes from one model to the other or how to change one model in such a way that propagation is guaranteed to be possible. 
Restricting model synchronization to these derived rules, however, may lead to unnecessary deletion and recreation of model elements during change propagation. 
This is inefficient and may cause unnecessary information loss, i.e., when deleted elements contain information that is not represented in the second model, this information cannot be recovered easily.
{\em Short-cut rules} have recently been developed to avoid unnecessary information loss by reusing existing model elements.
In this paper, we show how to automatically derive {\em (short-cut) repair rules} from short-cut rules to propagate changes such that information loss is avoided and model synchronization is accelerated. 
The key ingredients of our rule-based model synchronization process are these repair rules and an \emph{incremental pattern matcher} informing about suitable applications of them. 
We prove the termination and the correctness of this synchronization process and discuss its completeness. 
As a proof of concept, we have implemented this synchronization process in eMoflon, a state-of-the-art model transformation tool with inherent support of bidirectionality. 
Our evaluation shows that repair processes based on {\em (short-cut) repair rules} have considerably decreased information loss and improved performance compared to former model synchronization processes based on TGGs. 
\end{abstract}

\section{Introduction}
\label{sec:intro}
The close collaboration of multiple disciplines such as electrical engineering, mechanical engineering, and software engineering in system design often leads to discipline-spanning system models~\cite{GPR11}.
Keeping models synchronized by checking and preserving their consistency can be a challenging problem which is not only subject to ongoing research but also of practical interest for industrial applications. 
Model-based engineering has become an important technique to cope with the increasing complexity of modern software systems.
Various \emph{bidirectional transformation} (bx) approaches~\cite{CFHLST09,ACGMS18} for models have been suggested to deal with model (view) synchronization and consistency.
Across these different approaches the following are important research topics~\cite{GW09,DXC11,HPW12,MC16,CGMcKS17,HPC18,HB19}: \emph{incrementality}, i.e., achieving runtime/com\-ple\-xity dependent on the size of the model change, not on the model size, and \emph{least change}, i.e., keeping the resulting model as similar as possible to the original one while restoring consistency. 
In this work, we extend synchronization approaches based on \emph{triple graph grammars} by specific \emph{repair rules} to increase incrementality and efficiency and to decrease the amount of change that occurs during synchronization. 
We show how to avoid unnecessary information loss in model synchronization for scenarios \emph{in which one model is changed at a time}.
Throughout this paper we stick to this scenario of model synchronization.
The more general case of \emph{concurrent model synchronization} where both models have been altered is left to future work.

Triple Graph Grammars (TGGs) \cite{Schuerr95} are a declarative, rule-based bidirectional transformation approach, which allows to synchronize models of two different views (usually called the \emph{source} and \emph{target domain} in the TGG-related literature). 
The purpose of a TGG is to define a consistency relationship between pairs of models in a rule-based manner by defining traces between their elements. 
Given a TGG, 
its rules can be automatically \emph{operationalized} into \emph{source} and \emph{forward rules}.
While the source rules are used to build up models of the source domain, forward rules translate them to the target domain and thereby, establish traces between corresponding model elements. 
Analogously, target models can be propagated to the source domain by using \emph{target} and \emph{backward rules} that can be  automatically deduced as well. To avoid redundancy in presentation, we stick to forward propagation throughout this paper.

In \cite{Schuerr95}, a simple batch-oriented synchronization process was presented, which just re-translates the whole source model after each change using forward rules.
Several incremental synchronization processes based on TGGs have been presented in the literature thereafter.
A process is considered to be incremental if the target model is not recomputed from scratch but unaffected model parts are preserved as much as possible.\footnote{Ideally, the runtime (complexity) of a synchronization should depend on the size of the change to the source model and not on the sizes of the source and the target model~\cite{GW09}. This requirement is a good motivation for incremental synchronization.}
To obtain an incremental synchronization process, two basic strategies have been pursued (in combinations): 
(i) The synchronization algorithm takes additional information of forward rules into account. 
This information might consist of precedence relations over rules~\cite{LAVS12}, dependency information on model elements w.r.t. their creation~\cite{GW09,OP14}, a maximal, still consistent submodel~\cite{HEOCDXGE15}, or information about broken matches of forward rules provided by an incremental pattern matcher~\cite{LAFVS17,Leblebici18}. 
(ii) The actual propagation of changes in a synchronization process is not based on the application of forward rules exclusively but also uses additional rules. 
To propagate a deletion on the source part, almost all approaches support to revoke an application of a forward rule.
The recovation of rule applications is formalized as inverse rule applications in, e.g., \cite{LAVS12}. 
In addition, custom-made rules  have been used in synchronization algorithms that describe specific kinds of model edits in any modeling language~\cite{GH09} or in a concrete modeling language~\cite{BPDSD14}. Moreover, generalized forward rules have been defined which allow for re-use of elements~\cite{GH09,GPR11,OP14}. 
Summarizing, a number of approaches for incremental model synchronization based on TGGs have been presented in the literature.
Some of them such as \cite{GW09,GPR11} are informally presented without any guarantee to reestablish the consistency of modified models. Others present their synchronization approaches formally and show their correctness but are only applicable under restricted circumstances~\cite{HEOCDXGE15} or have not been implemented yet, such as \cite{OP14}.  Hence, {\em we still miss a TGG-based model synchronization approach that avoids unnecessary information loss, is proven to be correct, and is efficiently implemented}.

In this article, we present an incremental model synchronization approach based on an extended set of TGG rules.
In \cite{FKST18}, we introduced \shortcutRules{} for handling complex con\-sis\-ten\-cy-pre\-ser\-ving model updates while avoiding unnecessary information loss. 
A short-cut rule replaces one rule application with another one while preserving involved model elements (instead of deleting and re-creating them). 
We deduce source and forward rules from short-cut rules to support complex model edits and their synchronization with the target domain.

We present an incremental model synchronization algorithm based on short-cut rules and show its correctness. 
We implemented our synchronization approach in eMoflon~\cite{EL14,WAFVSL19,WARV19}, a state-of-the-art bidirectional model  transformation tool, and evaluate it.  
Being based on eMoflon, we are able to extend the synchronization process suggested by Leblebici (et al.)~\cite{LAFVS17,Leblebici18} and rely on information provided by an incremental pattern matcher also to detect when and where to apply our derived repair rules. 
However, the construction and derivation of these is general and could extend other suggested TGG-based synchronization processes as well. 
The results of our evaluation show that, compared to model synchronization in eMoflon without short-cut repair rules, the application of these repair rules allows to react to model changes in a less invasive way by preserving information. In addition, it shows more efficiency. 

This paper extends the work in~\cite{FKST19}.  
Beyond~\cite{FKST19}, we
\begin{itemize}
	\item present the actual synchronization process in pseudocode and prove its correctness and termination (based on the results obtained in~\cite{FKST19,LAFVS17,Leblebici18}), 
	\item extend our approach to deal with \emph{filter NACs} (a specific kind of negative application conditions in forward rules), 
	\item describe the implementation, especially the tool architecture, in more detail, 
	\item extend the evaluation by investigating the expressiveness of short-cut repair rules at the practical example of code refactorings~\cite{Fowler2018}, and
	\item consider the related work more comprehensively. 
\end{itemize}

The rest of this paper is organized as follows.
In Sect.~\ref{sec:example} we give an informal overview of 
our model synchronization approach. It shall allow readers to grasp the general idea without working through the technical details. 
In Sect.~\ref{sec:preliminaries} we recall triple graph grammars. Sect.~\ref{sec:sc-rules} recalls the construction of short-cut rules.
The construction of short-cut rules and their properties are presented in Sect.~\ref{sec:sc-rules}, while Sect.~\ref{sec:constructing-repair-rules} introduces the derivation of repair rules.  
Section~\ref{sec:syncProcess} focuses on the implemented synchronization algorithm and its formal properties.
To be understandable to readers who are not experts on algebraic graph transformation, we use a set-theoretical notion in these more technical sections, in contrast to the original contribution in~\cite{FKST19} which is based on category theory. 
Sect.~\ref{sec:implementation} describes the implementation of our model synchronization algorithm in eMoflon, focussing on the tool architecture.
Our synchronization approach is evaluated in Sect.~\ref{sec:implAndEvaluation}.
Finally, we discuss related work in Sect.~\ref{sec:related-work} and conclude with pointers to future work in Sect.~\ref{sec:conclusion}. 
Appendix~\ref{sec:evaluation-ruleset} presents the rule set used for our evaluation.

\section{Informal Introduction to TGG-Based Model Synchronization}
\label{sec:example}
In this section, we illustrate our approach to model synchronization. 
Using a simple example, we will explain the basic concepts as well as all main ingredients for our new synchronization process. 
Reading this section and having a passing view on the synchronization algorithm (Section~\ref{subsec:sync-process}), evaluation (Section~\ref{sec:implAndEvaluation}), and related works (Section~\ref{sec:related-work}) should give an adequate impression of the core ideas of our work. 

Graph transformations, and triple graph grammars in particular, are a suitable formal framework to reason about and to implement model transformations and synchronizations~\cite{BET12,EEGH15}.\footnote{Therefore, we will use the terms \enquote{graph} and \enquote{model} interchangeably in this paper. We will stick to the graph terminology in more formal sections.}
A triple graph consists of three graphs, namely the \emph{source}, \emph{target}, and \emph{correspondence graph}. 
The latter encodes which elements of source and target graph \emph{correlate to each other}.
This is done by mapping each element of the correspondence graph to an element of the source graph as well as to an element of the target graph (formally these are two \emph{graph morphisms}). 
Elements connected via such a mapping are considered to be correlated.

\emph{Triple graph grammars} (TGGs)~\cite{Schuerr95} declaratively define how consistent models co-evolve. 
This means that a triple graph is considered to be consistent if it can be derived from a start triple (e.g., the empty graph) using the rules of the given grammar. 
Furthermore, the rules can automatically be \emph{operationalized} to obtain new kinds of rules, e.g., for translation/synchronization processes.

We illustrate our model synchronization process by synchronizing a Java AST (abstract syntax tree) and a custom documentation model as example. 
This example has been basically introduced by Leblebici et al. \cite{Leblebici2017}; it is slightly modified to demonstrate the core concepts of our approach.
Note, however, that the evaluation in Sect.~\ref{sec:implAndEvaluation} is based on a larger and more complex TGG consisting of 24 rules (as presented in App.~\ref{sec:evaluation-ruleset}). 
\begin{figure}[ht]
	\centering
	\includegraphics[width=0.8\columnwidth]{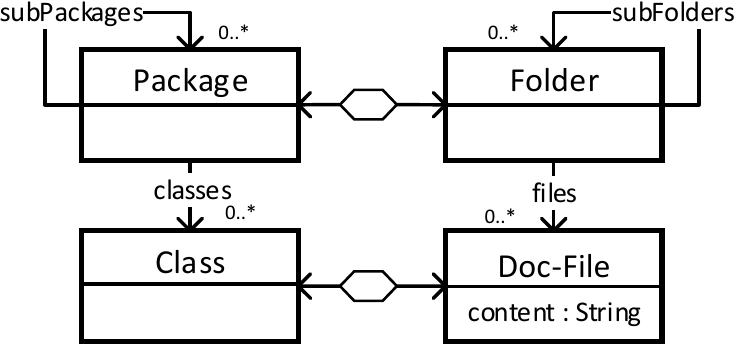}
	\caption{Example: Type Graph}
	\label{fig:typegraph}
\end{figure}

For model synchronization, we consider a Java AST mod\-el as {\em source} model and its documentation model as {\em target} model, i.e., changes in a Java AST model have to be transferred to its documentation model and vice versa.
Note that we do not consider concurrent model synchronization, i.e., concurrent changes to both sides that have to be synchronized. 
Figure~\ref{fig:typegraph} depicts the type graph that describes the syntax of our example triple graphs. 
It shows a \package{} hierarchy and \classes{} as the source side, a \folder{} hierarchy with \docs{} as target side and correspondence types in between depicted as hexagons.
Furthermore, \docs{} have an attribute \textsf{content} which is of type \textsf{String}.
Note that, in our example, there are two correspondence types which can be distinguished by the type of elements they connect on both sides.

\begin{figure}[ht]
	\centering
	\includegraphics[width=\columnwidth]{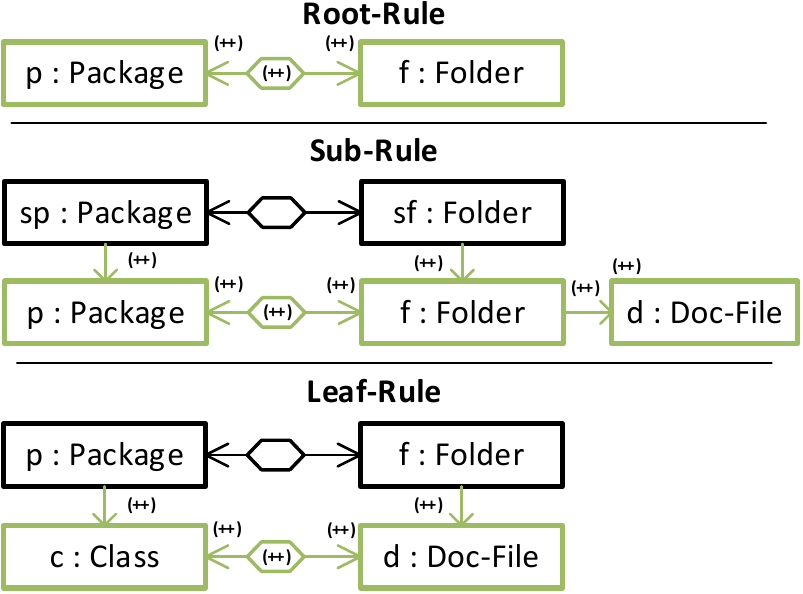}
	\caption{Example: TGG Rules}
	\label{fig:tggRules}
\end{figure}

\paragraph{TGG rules.}
Figure~\ref{fig:tggRules} shows the rule set of our example TGG consisting of three rules (assuming an empty start graph): 
\begin{figure*}
	\centering
	\includegraphics[width=1.0\textwidth]{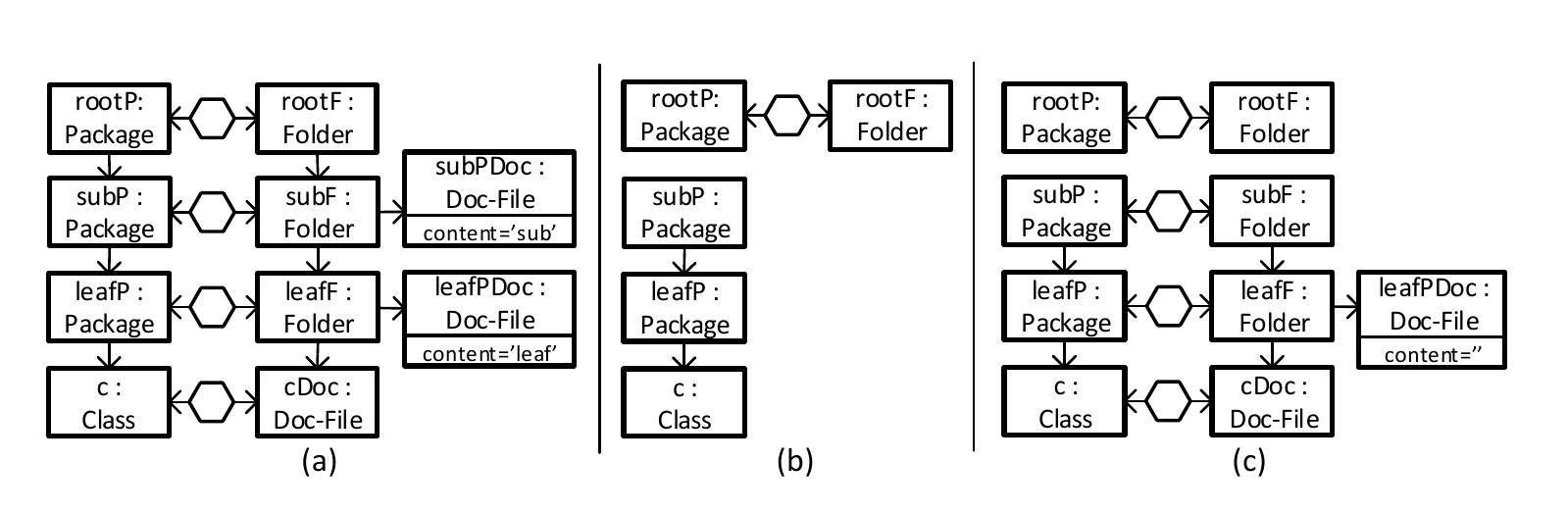}
	\caption{Exemplary Synchronization Scenario}
	\label{fig:translationExample}
\end{figure*}
\firstTGGRule{} creates a root \package{} together with a root \folder{} and a correspondence link in between.
This rule has an empty precondition and creates elements only; they are depicted in green and with the annotation (++). 
\secondTGGRule{} creates a \package{} and \folder{} hierarchy given that an already correlated \package{} and \folder{} pair exists.
Finally, \thirdTGGRule{} creates a \class{} and a \doc{} under the same precondition as \secondTGGRule{}.
 
TGG rules can be used to generate triple graphs; triple graphs generated by them are consistent by definition.
An example is depicted in Fig.~\ref{fig:translationExample}~(a) which can be generated by first applying \firstTGGRule{} followed by two applications of \secondTGGRule{} and an application of \thirdTGGRule{}: 
Starting with the empty triple graph, the first rule application just creates the elements \rootP{} and \rootF{} and the correspondence element in between. 
The second rule application matches these elements and creates \subP{}, \subF{}, \subPDoc{}, their respective incoming edges, and the correspondence element between \subP{} and \subF{}. 
The other two rule applications are performed similarly. 

\paragraph{Operationalization of TGG rules.} 
A TGG can also be used for \emph{translating} a model of one domain to a correlated model of a second domain. 
Moreover, a TGG offers support for \emph{model synchronization}, i.e., for restoring the consistency of a triple graph that has been altered on one side. 
For these purposes, each TGG rule has to be operationalized to two kinds of rules: A {\em source} rules enable changes of source models (e.g., as performed by a user) while forward rules translate such changes to the target model.\footnote{Analogously, \emph{target} and \emph{backward rules} can be derived.} 
The result of applying a source rule followed by an application of its corresponding forward rule yields the same result as applying the TGG rule they originate from.
Figure~\ref{fig:tggSourceRules} shows the resulting {\em source} rules for our example TGG.

\begin{figure}
	\centering		
	\includegraphics[width=0.75\columnwidth]{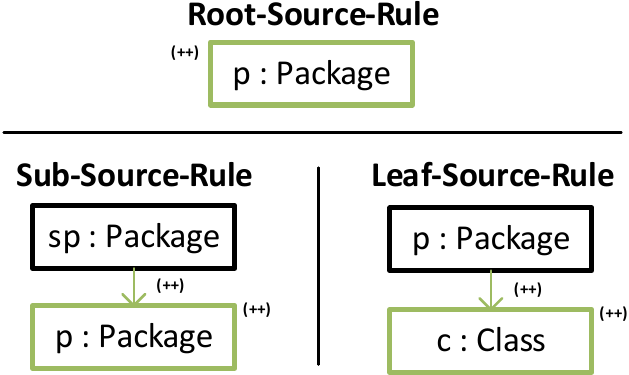}
	\caption{Example: TGG Source Rules}
	\label{fig:tggSourceRules}
\end{figure}

\begin{figure}
	\centering		
	\includegraphics[width=\columnwidth]{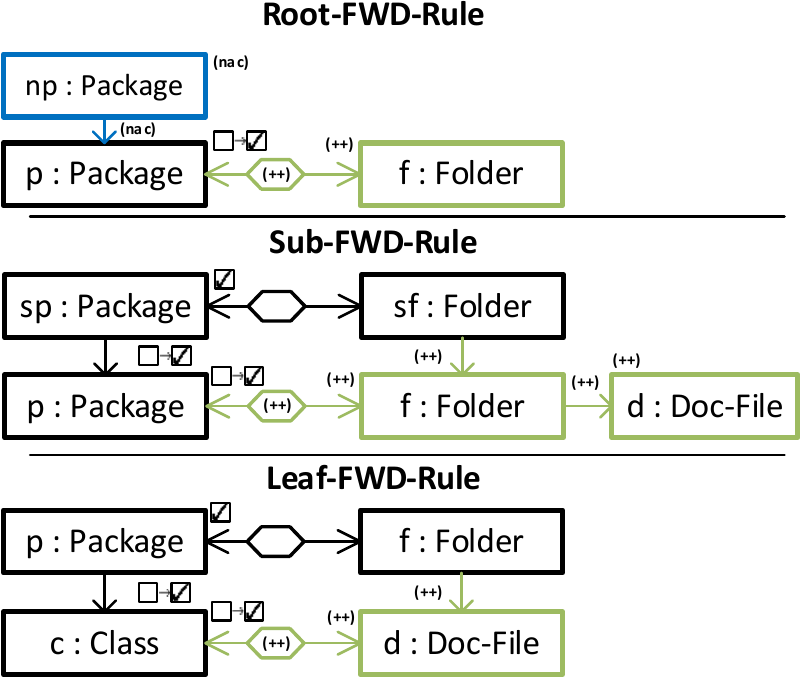}
	\caption{Example: TGG Forward Rules}
	\label{fig:tggFwdRules}
\end{figure}

\paragraph{Forward translation rules.}
Figure~\ref{fig:tggFwdRules} depicts the resulting {\em forward} rules.
They have a similar structure compared to their original TGG rules with three important differences.
First, elements on the source side are now considered as context and as such have to be matched as a precondition for this rule to be applicable.
Second, since we consider elements on the source side to already be present, we have to mark whether an element has already been translated or not.
A \marked{} annotation can be found on source elements which must have been translated before.
On the other hand, \marking{} annotations indicate that applying this rule would mark this element as translated.
This annotation can be found at elements that are created by the original TGG rule.
Possible formalizations of these marking are given, e.g., in \cite{Hermann2010,LAFVS17}.
The third difference is the use of negative application conditions (NACs)~\cite{EEPT06} which are indicated with a (nac) and depicted in blue.
Using NACs, we are able to not only define necessary structure that has to be found but also the explicit absence of structural elements as in \firstTGGForwardRule{} where we forbid \subP{} to have a parent package. 
The theory behind these so-called \emph{filter NACs} is formalized by Hermann et~al.~\cite{Hermann2010} and they can be derived automatically from the rules of a given TGG when computing its forward rules.

Using these rules, we can translate Java AST to documentation models. 
Considering the one on the source side of the triple graph in Fig.~\ref{fig:translationExample} (a), it is translated to a documentation model such that the result is the complete graph depicted in this part of the figure. 
To obtain this result we apply \firstTGGForwardRule{} at the root \package{}, \secondTGGForwardRule{} at \packages{} \subP{} and \leafP{}, and finally \thirdTGGForwardRule{} at \class{} \cl{}. 
Note that \secondTGGForwardRule{}, for example, is applicable when matching \packages{} \ssp{} and \p{} of the rule to the \packages{} \rootP{} and \subP{} of the source graph, respectively, since \rootP{} was marked as translated by the application of \firstTGGForwardRule{}. 
Without the NAC in \firstTGGForwardRule{}, this rule would also be applicable at the elements \subP{} and \leafP{}. 
Applying \firstTGGForwardRule{} and translating these elements with it, however, would result in the edges from their parent \packages{} not being translatable any longer: 
there is no rule in our TGG rule set that creates edges between packages only.
Hence, NACs can direct the translation process to avoid these dead-ends. 
Filter NACs are derived such that they prevent rule applications leading to dead-ends, only. 

\paragraph{Existing approaches to model synchronization.}
Given a triple graph such as the one in Fig.~\ref{fig:translationExample} (a), a developer may want to split the modeled project into multiple ones. 
For this purpose, a subpackage such as \subP{} shall become a root package.
Since \subP{} was created and translated as a sub package rather than a root element, this model change introduces an inconsistency.
To resolve this issue, the approaches presented in~\cite{GW09,LAVS12,LAFVS17,Leblebici18} and, to a certain degree, also the one in~\cite{HEOCDXGE15} revert the translation of \subP{} into \subF{} and re-translate \subP{} with an appropriate translation rule such as \firstTGGForwardRule{}.
Reverting the former translation step may lead to further inconsistencies as we remove elements that were needed as context elements by other applications of forward rules.
The result is a reversion of all translation steps except for the first one which translates the original root element. 
The result is shown in Fig.~\ref{fig:translationExample} (b).
Thereafter, the untranslated elements can be re-translated yielding the result graph in (c). 
This example shows that this synchronization approach may delete and re-create a lot of similar structures which appears to be inefficient. 
Second, it may lose information that exists on the target side only, e.g., documentation saved in the \textsf{content} attribute which is empty now as it cannot be restored from the source side only. 
Such an information loss is unnecessary as we will show below. 
Instead of deleting elements and recreating them, we will present a synchronization process that aims to preserve information as much as possible. 

\paragraph{Model synchronization with short-cut repair.}
In~\cite{FKST18}, we introduce \shortcutRules{} as a kind of sequential rule composition mechanism that allows to replace one rule application with another one while elements are preserved (instead of deleted and recreated). 

\begin{figure}[ht]
	\centering
	\includegraphics[width=\columnwidth]{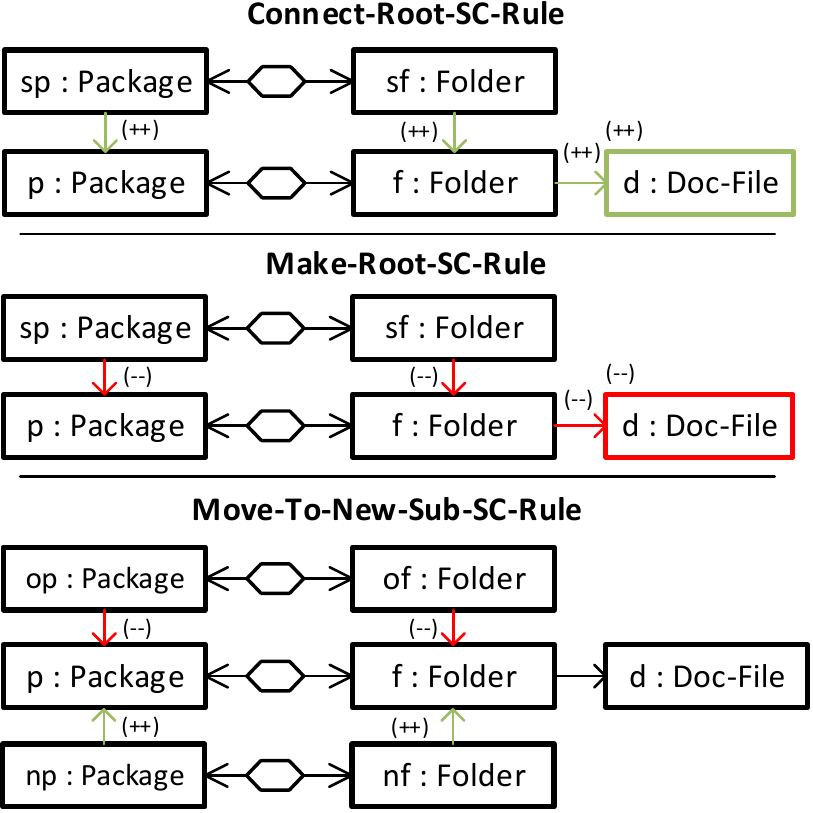}
	\caption{Short-cut rules}
	\label{fig:scRules}
\end{figure} 

Figure~\ref{fig:scRules} depicts three \shortcutRules{} which can be derived from our original three TGG rules.
The first two, \firstSCRule{} and \secondSCRule{}, are derived from \firstTGGRule{} and \secondTGGRule{}. 
The upper \shortcutRule{} replaces an application of \firstTGGRule{} with one of \secondTGGRule{} and turns root elements into sub elements.
In contrast, the lower \shortcutRule{} replaces  an application of \secondTGGRule{} with one of \firstTGGRule{}, thus, turning sub elements into root elements. 
Both \shortcutRules{} preserve the model elements present in their corresponding TGG rules and solely create elements that do not exist yet (++), or delete those depicted in red and annotated with (-\--) which became superfluous.
The third \shortcutRule{} \thirdSCRule{} relocates sub elements and replaces a \secondTGGRule{} application with another one of the same kind.

A \shortcutRule{} is constructed by overlapping two rules with each other where the first one is the replaced and the second the replacing rule.
Overlapped elements are preserved such as \p{} and \f{} in \firstSCRule{}.
Created elements that are not overlapped fall into two categories.
If the element was created in the replaced rule but is superfluous in the replacing rule, it is deleted, e.g, \dd{} in \secondSCRule{}.
On the other hand, if the element was not created by the replaced rule but by the replacing rule, then the element is created, e.g., \dd{} in \firstSCRule{}.
Context elements can be mapped as well while unmapped context elements from both rules are glued onto the final \shortcutRule{}, e.g., \op{} and \of{} which are context in the replaced rule, and \np{} and \nf{} which are context in the replacing rule.  
Since there are many possible overlaps for each pair of rules, constructing a reasonable set of \shortcutRules{} depends on the concrete example TGG and the requirement for advanced model changes that go beyond the standard capabilities of TGG based model synchronizers.
Usually, it is worthwhile to construct \shortcutRules{} for frequent model changes in order to increase the synchronization efficiency and decrease information loss in these cases. 

In our example above, the user wants to transform the triple graph in Fig.~\ref{fig:translationExample} (a) to the one in (c). Using \secondSCRule{} and matching the \packages{} \texttt{sp} and \texttt{p} to the \packages{} \rootP{} and \subP{} 
of the model (a) (and the correspondence nodes and \folders{} accordingly), this transformation is performed with a single rule application. 
Analogously, the triple graph (c) can be directly transformed backwards to (a) using \firstSCRule{}. 
Thus, these rules allow for complex user-edits on both, source and target side; they preserve the consistency of the model. 
However, there are also scenarios where applying a \shortcutRule{} may lead to an inconsistent state of the resulting triple graph.
A simple example is that of applying \firstSCRule{} in order to connect \texttt{subP} and \texttt{subF} with \texttt{rootP} and \texttt{rootF}, respectively.
The result would be a cycle in both, the \package{} and the \folder{} hierarchies; this model is no longer in the language of our example TGG.
In Sect.~\ref{sec:sc-rules}, we present sufficient conditions for the application of \shortcutRules{} to avoid such cases. 

\begin{figure}
	\centering
	\includegraphics[width=\columnwidth]{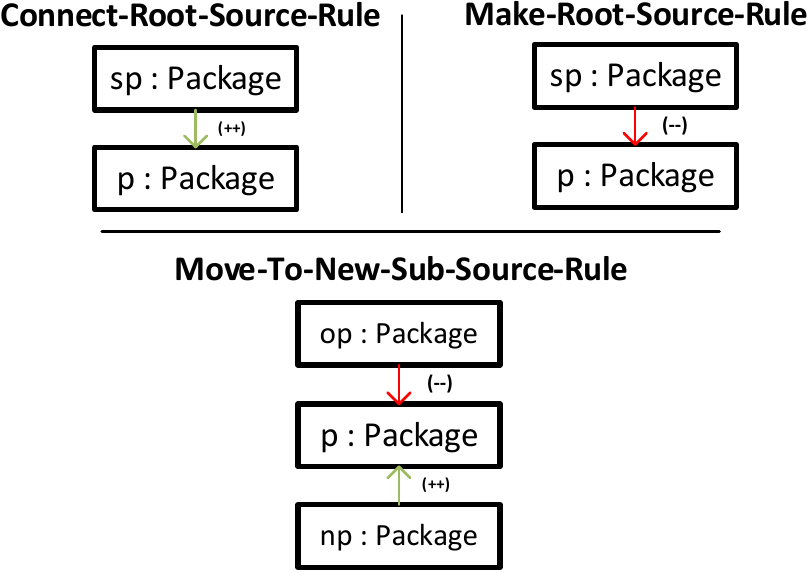}
	\caption{Short-cut Source Rules}
	\label{fig:SCRules_source}
\end{figure}

\begin{figure}
	\centering
	\includegraphics[width=\columnwidth]{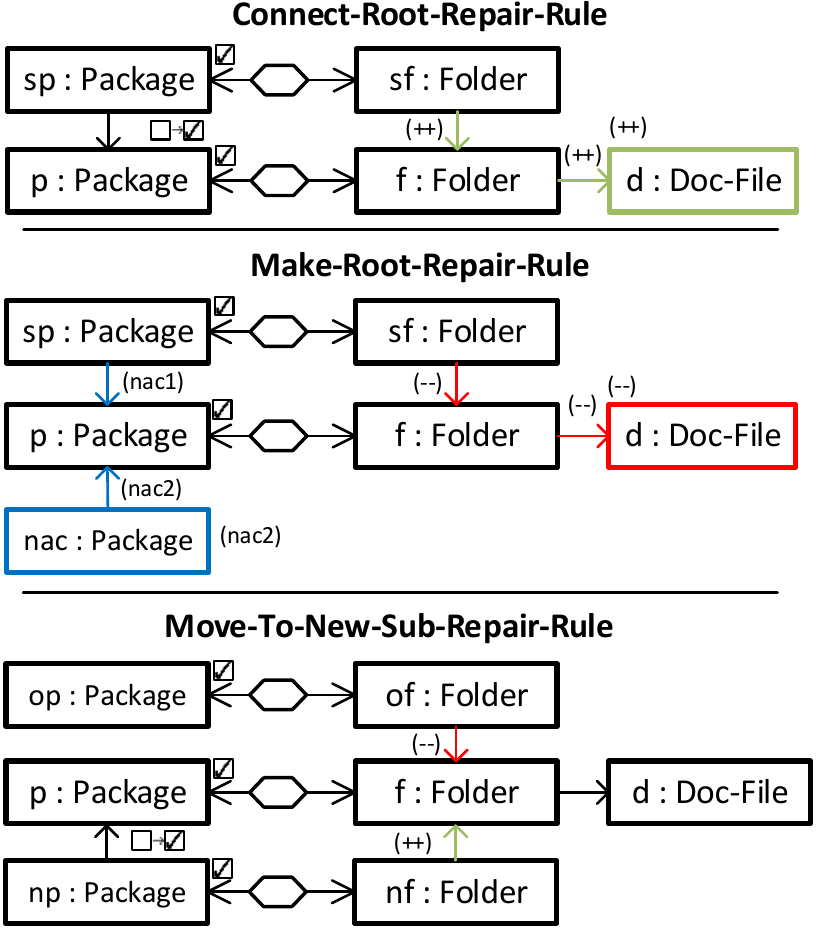}
	\caption{Repair rules}
	\label{fig:fwdSCRules}
\end{figure}

\paragraph{Operationalization of \shortcutRules{}.}
\ShortcutRules{} transform both models at once as TGG rules usually do and therefore, they cannot cope with the change of a single model.
Hence, similar to TGG rules, we have to operationalize them, thereby obtaining {\em short-cut source} and {\em short-cut repair} rules. 
Figure~\ref{fig:SCRules_source} depicts the {\em short-cut source} rules which are derived analogously to those of standard TGG rules.
In order to be able to handle the deleted edge between \rootP{} and \subP{}, as deleted by \secondSCSourceRule{}, for example, a repair rule is needed that adapts the target graph accordingly by deleting the now superfluous edge between \rootF{} and \subF{}.
Figure~\ref{fig:fwdSCRules} depicts the resulting repair rules derived from the short-cut rules in Fig.~\ref{fig:scRules}. 
A \shortcutRule{} is forward operationalized by removing deleted elements from the rule\rq{}s source graph since these deletions have already happened.
Furthermore, created source elements become context because we expect them to already exist, e.g., through the a prior source rule application.
Finally, since \shortcutRules{} transform an application of one rule into that of another, filter NACs are added during operationalization to comply with application conditions of the replacing rule which naturally have to hold when applying the \shortcutRule{}.
Hence, \secondRRule{} is only applicable and can turn \subF{} into a root \folder{} if \subP{} has no parent packages and, thus, is indeed a root \package{} itself.
Note that \firstTGGForwardRule{} is only applicable if \subP{} has no parent packages, which \secondRRule{} has to incorporate as well.
For this reason, \secondRRule{} contains \textit{nac1}, which forbids \rootP{} to be the parent package of \subP{} and \textit{nac2}, which forbids \subP{} to have any other parent packages than \rootP{}.

{\em Short-cut repair} rules allow to propagate graph changes directly to the other graph to restore consistency. 
Revisiting our example of Fig.~\ref{fig:translationExample}, we are now able to use \secondRRule{} to propagate the deleted edge between \subP{} and \rootP{} by deleting the corresponding edge between \subF{} and \rootF{} and the now superfluous \doc{} \subPDoc{}. 
The result is the consistent triple graph again depicted in Fig.~\ref{fig:translationExample} (c) with the content attribute of \leafPDoc{} containing the value \enquote*{leaf}.
So, this repair does not cause information loss and allows to skip the costly reversion process with the intermediate result in Fig.~\ref{fig:translationExample} (b). 

Summarizing, the user edit of removing the edge between \packages{} \rootP{} and \subP{} corresponds to the source rule of \secondSCRule{}, namely \secondSCSourceRule{}, and the according update to the target side is performed by \secondRRule{} which is the corresponding repair rule.
Together, they perform an edit step structurally equivalent to the one depicted by the triple graphs in Fig.~\ref{fig:translationExample} (a) and (c); however, the value of the attribute \textsf{content} does not get lost. 
Alternatively, this step can be obtained by applying the short-cut rule \secondSCRule{}.
This is not a coincidence: 
In~\cite[Theorem~7]{FKST19}, we showed that applying the source rule of a short-cut rule (which corresponds to a user edit on the source part only) followed by an application of the corresponding repair rule at the according match is the same as applying the original short-cut rule.

\section{Preliminaries: Triple Graphs, Triple Graph Grammars and their Operationalizations}
\label{sec:preliminaries}
In this section, we recall triple graph grammars (TGGs) and their operationalization~\cite{Schuerr95}. 
Our derivation of repair rules is based on the construction of so-called short-cut rules~\cite{FKST18}, which we recall as well. 
For simplicity, we stick with set-theoretic definitions of the involved concepts (in contrast to category-theoretic ones as, e.g., in~\cite{EEPT06,EEGH15,FKST18,FKST19}). 
Moreover, while we provide formal definitions for central notions, we will just explain others and provide references for their formal definitions.  

\subsection{Graphs, triple graphs, and their transformations}
\emph{Graphs} and their (rule-based) \emph{transformations} are suitable to formalize various kinds of models and their evolution, in particular of EMF models~\cite{BET12}.\footnote{Therefore in this paper, we use the terms \emph{graph} and \emph{model} interchangeably. In the formal parts, we will consequently speak of graphs following the formal literature.}
In the context of this work, a \emph{graph} consists of a set of nodes and a set of directed edges which connect nodes. 
Graphs may be related by \emph{graph morphisms}, and a \emph{triple graph} consists of three graphs connected by two graph morphisms. 

\begin{definition}[Graph, graph morphism, triple graph, and triple graph morphism]\label{def:graphs}
	A \emph{graph} $G = (V,E,s,t)$ consists of a set $V$ of vertices, a set $E$ of edges, and source and target functions $s,t: E \rightarrow V$.
	An {\em element} $x$ of $G$ is a node or an edge, i.e., $x \in V$ or $x \in E$.
	A \emph{graph morphism} $f:G \rightarrow H$ between graphs $G = (V_G, E_G, s_G, t_G)$ and $H = (V_H, E_H, s_H, t_H)$ consists of two functions $f_V: V_G \rightarrow V_H$ and $f_E: E_G \rightarrow E_H$ that are compatible with the assignment of source and target to edges, i.e., $f_V \circ s_G = s_H \circ f_E$ and $f_V \circ t_G = t_H \circ f_E$. 
	Given a fixed graph $\mathit{TG}$, a \emph{graph typed over $\mathit{TG}$} is a graph $G$ together with a graph morphism $\mathit{type}_G: G \to \mathit{TG}$.
	A \emph{typed graph morphism} $f: (G, \mathit{type}_G) \to (H, \mathit{type}_H)$ between typed graphs is a graph morphism $f: G \to H$ that respects the typing, i.e., $\mathit{type}_G = \mathit{type}_H \circ f$ (componentwise).
	A (typed) graph morphism $f = (f_V,f_E)$ is \emph{injective} if both $f_V$ and $f_E$ are. 
	
	A \emph{triple graph} $G = (G_S \xleftarrow{\sigma_G} G_C \xrightarrow{\tau_G} G_T)$ consists of three graphs $G_S, G_C, G_T$, called \emph{source, correspondence}, and \emph{target graph}, and two graph morphisms $\sigma_G: G_C \rightarrow G_S$ and $\tau_G: G_C \rightarrow G_T$, called \emph{source} and \emph{target correspondence morphism}. 
	A \emph{triple graph morphism} $f: G \rightarrow H$ between two triple graphs $G$ and $H$ consists of three graph morphisms $f_S: G_S \rightarrow H_S, f_C: G_C \rightarrow H_C$ and $f_T: G_T \rightarrow H_T$ such that $\sigma_H \circ f_C = f_S \circ \sigma_G$ and $\tau_H \circ f_C = f_T \circ \tau_G$. 
	Given a fixed triple graph $\mathit{TG}$, a \emph{triple graph typed over $\mathit{TG}$} is a triple graph $G$ together with a triple graph morphism $\mathit{type}_G: G \to \mathit{TG}$. 
	Again, \emph{typed triple graph morphisms} are triple graph morphisms that respect the typing.
	A (typed) triple graph morphism $f = (f_S,f_C,f_T)$ is \emph{injective} if $f_S, f_C$, and $f_T$ all are.
\end{definition}

\begin{example}
	Figure~\ref{fig:translationExample} depicts three triple graphs; their common type graph is depicted in Fig.~\ref{fig:typegraph}. 
	The typing morphism is indicated by annotating the elements of the triple graphs with the types to which they are mapped in the type graph. 
	The nodes in the triple graphs are of types \package{}, \folder{}, \class{}, and \doc{}. 
	In each case, the source graph is depicted to the left and the target graph to the right. 
	The hexagons in the middle constitute the correspondence graphs. 
	Formally, the edges from the correspondence graphs to source and target graphs are morphisms: 
	The edges encode how an individual correspondence node is mapped by the correspondence morphisms. 
	For example, the nodes \rootP{} and \rootF{} of types \package{} and \folder{} correspond to each other as they share the same correspondence node as preimage under the correspondence morphisms. 
\end{example}

Rules offer a declarative means to specify transformations of (triple) graphs. 
While classically rewriting of triple graphs has been performed using non-deleting rules only, we define a less restricted notion of rules\footnote{As used in double pushout rewriting of graphs or objects of other adhesive categories more generally~\cite{EEPT06,LS05}.} right away since short-cut rules and repair rules derived from them are both potentially deleting.  
A rule $p$ consists of three triple graphs, namely a \emph{left-hand side} (LHS) $L$ and a \emph{right-hand side} (RHS) $R$ and an \emph{interface} $K$ between them. 
Applying such a rule to a triple graph $G$ means to choose an injective morphism $m$ from $L$ to $G$. 
The elements from $m(L\setminus l(K))$ are to be deleted; if this results in a triple graph again, the morphism $m$ is called a match and $p$ is applicable at that match. 
After this deletion, the elements from $R \setminus r(K)$ are added; the whole process of applying a rule is also called a \emph{transformation (step)}. 

\begin{definition}[Rule, transformation (step)]
	A \emph{rule} $p = (L \xleftarrow{l} K \xrightarrow{r} R)$ consists of three triple graphs, $L$, $R$, and $K$, called the \emph{left-hand side}, \emph{right-hand side}, and \emph{interface}, respectively, and two injective triple graph morphisms $l: K \to L$ and $r: K \to R$. 
	A rule is called \emph{monotonic}, or \emph{non-deleting}, if $l$ is an isomorphism. 
	In this case we denote the rule as $r: L \to R$. 
	The \emph{inverse rule} of a rule $p$ is the rule $p^{-1} = (R \xleftarrow{r} K \xrightarrow{l} L)$. 
	
	Given a triple graph $G$, a rule $p = (L \xleftarrow{l} K \xrightarrow{r} R)$, and an injective triple graph morphism $m: L \to G$, the rule $p$ is \emph{applicable at $m$} if 
	\begin{equation*}
		D \coloneqq G \setminus (m(L \setminus l(K))) \enspace ,
	\end{equation*}
	is a triple graph again. Operator $\setminus$ is understood as node- and edge-wise set-theoretic difference. The source and target functions of $D$ are restricted accordingly. 
	If $D$ is a triple graph, 
	\begin{equation*}
		H \coloneqq D \cup n(R \setminus r(K)) \enspace ,
	\end{equation*}
	is computed. Operator $\cup$ is understood as node- and edge-wise set-theoretic union.  
$n(R \setminus r(K))$ is a new copy of newly created elements. $n$ can be extended to $R$ by $n(r(K)) = m(l(K))$.
The values of the source and target functions for edges from $n(R \setminus r(K))$ with source or target node in $K$ are determined by $m \circ l$, i.e.,
	\begin{align*}
		s_H(e) & \coloneqq m(l(r^{-1}(s_R(e)))) \\
		t_H(e) & \coloneqq m(l(r^{-1}(t_R(e)))) 
	\end{align*}
	for such edges $e \in n(E_R)$ with $s_R(e) \in r_V(V_K)$ or $t_R(e) \in r_V(V_K)$.
	The whole computation is called a \emph{transformation (step)}, denoted as $G \Rightarrow_{p,m} H$ or just $G \Rightarrow H$, $m$ is called a \emph{match}, $n$ is called  a \emph{comatch} and $D$ is the \emph{context triple graph} of the transformation.
\end{definition}
An equivalent definition based on computing two \emph{pushouts}, a notion from category theory generalizing the union of sets along a common subset, serves as basis when developing a formal theory~\cite{EEPT06}. 
In the following and in our examples, we always assume $K$ to be a common subgraph of $L$ and $R$ and the injective morphisms $l$ and $r$ to be the corresponding inclusions; this significantly eases the used notation. 
When we talk about the union of two graphs $G_1$ and $G_2$ along a common subgraph $S$, we assume that $G_1 \cap G_2 = S$. 

To enhance expressiveness, a rule may contain \emph{negative application conditions} (NACs)~\cite{EEPT06}. 
A NAC extends the LHS of a rule with a forbidden pattern: 
A rule is allowed to be applied only at matches which cannot be extended to any pattern forbidden by one of its NACs. 
If we want to stress that a rule is not equipped with NACs, we call it a \emph{plain rule}. 

\begin{definition}[Negative application conditions]
	Given a rule $p = (L \leftarrow K \rightarrow R)$, a set of \emph{negative application conditions} (NACs) for $p$ is a finite set of graphs $\mathit{NAC} = \{N_1, \dots, N_k\}$ such that $L$ is a subgraph of every one of them, i.e., $L \subset N_i$ for $1 \leq i \leq k$.
	
	A rule $(p = (L \leftarrow K \rightarrow R),\mathit{NAC})$ with NACs is applicable at a match $m: L \rightarrow G$ if the plain rule $p$ is and, moreover, for none of the NACs $N_i$ there exists an injective morphism $x_i: N_i \to G$ such that $x_i \circ \iota_i = m$ where $\iota_i: L \hookrightarrow N_i$ is the inclusion of $L$ into $N_i$.
\end{definition}

\begin{example}
	Different sets of triple rules are depicted in Figs.~\ref{fig:tggRules}, \ref{fig:tggFwdRules}, \ref{fig:scRules}, and \ref{fig:fwdSCRules}. 
	All rules in these figures are presented in an \emph{integrated} form: 
	Instead of displaying LHS, RHS, and the interface as three separate graphs, just one graph is presented where the different roles of the elements are displayed using markings (and color). 
	The unmarked (black) elements constitute the interface of the rule, i.e., the context that has to be present to apply a rule. 
	Unmarked elements and elements marked with $(--)$ (black and red elements) form the LHS while unmarked elements and elements marked with $(++)$ (black and green elements) constitute the RHS. 
	Elements marked with (nac) (blue elements) extend the LHS to a NAC; different NACs for the same rule are distinguished using names. 
	As triple rules are depicted, their LHSs and RHSs are triple graphs themselves. 
	For example, the LHS $L$ of \secondTGGRule{} (Fig.~\ref{fig:tggRules}) consists of the nodes \ssp{} and \ssf{} of types \package{} and \folder{} and the correspondence node in between. 
	
	While, e.g., all rules in Fig.~\ref{fig:tggRules} are monotonic, \secondSCRule{} is not as it deletes edges and a \doc{}. 
	Applying \secondSCRule{} to the triple graph (a) in Fig.~\ref{fig:translationExample} leads to the triple graph (c), when \package{}-nodes \ssp{} and \p{} (of the rule) are matched to \rootP{} and \subP{} (in the graph), respectively. (The \folders{} on the target part are mapped accordingly.)
	The rules \firstSCRule{} and \secondSCRule{} are inverse to each other. 
	
	Finally, \firstTGGForwardRule{} (Fig.~\ref{fig:tggFwdRules}) depicts a rule that is equipped with a NAC: 
	It is applicable only at \packages{} that are not referenced by other \packages{}. 
	This means that it is applicable at node \subP{} in the triple graph (b) depicted in Fig.~\ref{fig:translationExample}, but not at node \leafP{}. 
\end{example}

\subsection{Triple graph grammars and their operationalization} 
Sets of triple graph rules can be used to define languages.

\begin{definition}[Triple graph grammar]
	A \emph{triple graph grammar} (TGG) $\mathit{GG} = (\mathcal{R},S)$ consists of a set of plain, monotonic triple rules $\mathcal{R}$ and a start triple graph $S$. 
	In case of typing, all rules of $\mathcal{R}$ and $S$ are typed over the same triple graph. 
	
	The language of a TGG $\mathit{GG}$, denoted as $\mathcal{L}(\mathit{GG})$, is the reflexive and transitive closure of the relation induced by transformation steps via rules from $\mathcal{R}$, i.e., 
	\begin{equation*}
		\mathcal{L}(\mathit{GG}) \coloneqq \{H \, | \, S \Rightarrow_{\mathcal{R}}^* H\}
	\end{equation*}
	where $\Rightarrow_{\mathcal{R}}^*$ denotes a finite sequence of transformation steps where each rule stems from $\mathcal{R}$. 
	
	The \emph{projection of the language of a TGG to its source part} is the set
	\begin{equation*}
		\mathcal{L}_S(\mathit{GG}) \coloneqq \{ G_S \, | \, G = (G_S \leftarrow G_C \rightarrow G_T) \in \mathcal{L}(\mathit{GG})\} \enspace ,
	\end{equation*}
	i.e., it consists of the source graphs of the triple graphs of $\mathcal{L}(\mathit{GG})$.
\end{definition}

In applications, quite frequently, the start triple graph of a TGG is just the empty triple graph. 
We use $\emptyset$ to denote the empty graph, the empty triple graph, and morphisms starting from the empty (triple) graph; it will always be clear from the context what is meant.
To enhance expressiveness of TGGs,  their rules can be extended with NACs or with some attribution concept for the elements of  generated triple graphs. 
A recent overview of such concepts and their expressiveness can be found in~\cite{WOR19}. 
In the following, we first restrict ourselves to TGGs that contain plain rules only and discuss extensions of our approach subsequently. 

\begin{example}
	The rule set depicted in Fig.~\ref{fig:tggRules}, together with the empty triple graph as start graph, constitutes a TGG. 
	The triple graphs (a) and (c) in Fig.~\ref{fig:translationExample} are elements of the language defined by that grammar while the triple graph (b) is not.
\end{example}

The operationalization of triple graph rules into \emph{source} and \emph{forward} (or, analogously, into \emph{target} and \emph{backward}) \emph{rules} is central to working with TGGs.
Given a rule, its source rule performs the rule\rq{}s actions on the source graph only while its forward rule propagates these to correspondence and target graph. 
This means that, for example, source rules can be used to generate the source graph of a triple graph while forward rules are then used to translate the source graph to correspondence and target side such that the result is a triple graph in the language of the TGG. 
Classically, this operationalization is defined for monotonic rules only~\cite{Schuerr95}. We will later explain how to extend it to arbitrary triple rules. 
We also recall the notion of marking~\cite{Leblebici18} and \emph{consistency patterns} which can be used to check if a triple graph belongs to a given TGG.

\begin{definition}[Source and forward rule. Consistency pattern]
	Given a plain, monotonic triple rule $r = L \rightarrow R$ with $r = (r_S, r_C, r_T)$, $L = (L_S \xleftarrow{\sigma_L} L_C \xrightarrow{\tau_L} L_T)$ and $R = (R_S \xleftarrow{\sigma_R} R_C \xrightarrow{\tau_R} R_T)$, its \emph{source rule} is defined as
	\begin{equation*}
		r^S \coloneqq (L_S \leftarrow \emptyset \rightarrow \emptyset) \xrightarrow{(r_S, \mathit{id}_{\emptyset}, \mathit{id}_{\emptyset})} (R_S \leftarrow \emptyset \rightarrow \emptyset) \enspace .
	\end{equation*}
	Its \emph{forward rule} is defined as 
	\begin{equation*}
		r^F \coloneqq (R_S \xleftarrow{\sigma_R \circ r_C} L_C \xrightarrow{\tau_L} L_T) \xrightarrow{(\mathit{id}_{R_S}, r_C, r_T)} (R_S \xleftarrow{\sigma_R} R_C \xrightarrow{\tau_R} R_T) \enspace .
	\end{equation*}
	We denote the left- and right-hand sides of source and forward rules of a rule $r$ by $L^S,L^F,R^S$, and $R^F$, respectively. 
	
	The \emph{consistency pattern} derived from $r$ is the rule
	\begin{equation*}
		r^C \coloneqq (R_S \xleftarrow{\sigma_R} R_C \xrightarrow{\tau_R} R_T) \xrightarrow{(\mathit{id}_{R_S}, \mathit{id}_{R_C}, \mathit{id}_{R_T})} (R_S \xleftarrow{\sigma_R} R_C \xrightarrow{\tau_R} R_T)
	\end{equation*}
	that, upon application, just checks for the existence of the RHS of the rule without changing the instance it is applied to. 
	
	Given a rule $r$, each element $x \in R_S \setminus L_S$  is called a \emph{source marking element} of the forward rule $r^F$; each element of $L_S$ is called \emph{required}. 
	Given an application $G \Rightarrow_{r^F,\mathit{m^F}} H$ of a forward rule $r^F$, the elements of $G_S$ that have been matched by source marking elements of $r^F$, i.e., the elements of the set $\mathit{m}^F(R_S \setminus L_S)$ are called \emph{marked elements}.
	A transformation sequence
	\begin{equation}\label{eq:forward-sequence}
		G_0 \Rightarrow_{m_1^F,r_1^F} G_1 \Rightarrow_{m_2^F,r_2^F} \dots \Rightarrow_{m_t^F,r_t^F} G_t
	\end{equation}
	is called \emph{creation preserving} if no two rule applications in sequence (\ref{eq:forward-sequence}) mark the same element.
	It is called \emph{context preserving} if, for each rule application in sequence (\ref{eq:forward-sequence}), the required elements have been marked by a previous rule application in sequence (\ref{eq:forward-sequence}). 
If these two properties hold for sequence~(\ref{eq:forward-sequence}), it is called {\em consistently marking}.
	It is called \emph{entirely marking} if every element of the common source graph $G_S$ of the triple graphs of this sequence is marked by a rule application in sequence (\ref{eq:forward-sequence}).
\end{definition}

The most important formal property of this operationalization is that applying a (sequence of) source rule(s) followed by applying the (sequence of) corresponding forward rule(s) yields the same result as applying the (sequence of) original TGG rule(s) assuming consistent matches~\cite{Schuerr95,EEEHT07}. 

Moreover, there is a correspondence between triple graphs belonging to the language of a given TGG and consistently and entirely marking transformation sequences via its forward rules.
We formally state this correspondence as it is an ingredient for the proof of correctness of our synchronization algorithm.
\begin{lemma}[{see \cite[Fact~1]{LAFVS17} or ~\cite[Lemma~4]{Leblebici18}}]\label{lem:emccp-series}
	Let a TGG $\mathit{GG}$ be given. 
	There exists a triple graph $G = (G_S \leftarrow G_C \rightarrow G_T) \in \mathcal{L}(\mathit{GG})$ if and only if there exists a transformation sequence like the one depicted in~(\ref{eq:forward-sequence}) via forward rules from $\mathit{GG}$ such that $G_0 = (G_S \leftarrow \emptyset \rightarrow \emptyset),\ G_t = (G_S \leftarrow G_C \rightarrow G_T)$, and the transformation sequence is consistently and entirely marking. 
\end{lemma}

For practical purposes, forward rules and consistency patterns may be equipped with so-called \emph{filter NACs} which can be automatically derived from the set of rules of the given TGG. 
The simplest examples of such filter NACs arise through the following analysis:
For each rule that translate a node without translating adjacent edges it is first checked if other rules translate the same type of node but also translate an adjacent edge of some type. 
If this is the case, it is checked if there are further rules which only translate the detected kind of adjacent edge. 
If none is found, the original rule is equipped with a NAC forbidding the respective kind of edges. 
This avoids a dead-end in translation processes: 
In the presence of such a node with its adjacent edge, using the original rule to only translate the node leaves an untranslatable edge behind.
The filter NAC of \firstTGGForwardRule{} is derived in exactly this way. 
For the exact and more sophisticated derivation processes of filter NACs, we refer to the literature~\cite{Hermann2010,KLKS10}. 
For our purposes it suffices to recall their distinguishing property: 
Filter NACs do not prevent \enquote{valid} transformation sequences of forward rules. 
We state this property in the terminology of our paper.
\begin{fact}[{\cite[Fact~4]{Hermann2010}}]
	Given a TGG $\mathit{GG} = (\mathcal{R},S)$, for each $r \in \mathcal{R}$, let $r^{\mathit{FN}}$ denote the corresponding forward rule that is additionally equipped with a set of derived filter NACs. (This set might be empty). 
	For $G_0 = (G_S \leftarrow \emptyset \rightarrow \emptyset)$,
	there exists a consistently and entirely marking 
transformation sequence 
	\begin{equation*}
		G_0 \Rightarrow_{r_1^F,m_1^F} G_1 \Rightarrow_{r_2^F,m_2^F} \dots \Rightarrow_{r_t^F,m_t^F} G_t
	\end{equation*}
	via the forward rules (without filter NACs) derived from $\mathcal{R}$ if and only if the sequence
	\begin{equation*}
		G_0 \Rightarrow_{r_1^{\mathit{FN}},m_1^F} G_1 \Rightarrow_{r_2^{\mathit{FN}},m_2^F} \dots \Rightarrow_{r_t^{\mathit{FN}},m_t^F} G_t
	\end{equation*}
	exists, i.e, if none of the filter NACs blocks one of the above rule applications.
\end{fact}

\begin{example}
	The source rules of the triple rules depicted in Fig.~\ref{fig:tggRules} are depicted in Fig.~\ref{fig:tggSourceRules}. 
	They allow to create \packages{} and \classes{} on the source side without changing correspondence and target graphs. 
	The formally existing empty graphs at correspondence and target sides are not depicted. 
	The corresponding forward rules are given in Fig.~\ref{fig:tggFwdRules}. 
	Their required elements are annotated with \marked{} and their source marking elements with \marking{}. 
	The rule \firstTGGForwardRule{} is equipped with a filter NAC: 
	The given grammar does not allow to create a \package{} that is contained in another one with its original rule \firstTGGRule{}. 
	Hence, the derived forward rule should not be used to translate a \package{}, which is contained in another one, to a \folder{}.
	As evident in the examples, the application of a source rule followed by the application of the corresponding forward rule amounts to the application of the original triple rule if matched consistently. 
	
	The consistency patterns that are derived from the TGG rules of our example are depicted in Fig.~\ref{fig:consistency}. 
	They just check for existence of the pattern that occurs after applying the original TGG rule. 
	A consistency pattern is equipped with the filter NACs of both its corresponding forward and backward rule.
	In our example, only \emph{Root-Consistency-Pattern} receives such NACs; one from \firstTGGForwardRule{} and the second one from the analogous backward rule.
	An occurrence of a consistency pattern in our example model indicates that a specific location corresponds to a concrete TGG rule application. 
	Hence, a disappearance of such a match indicates that a former intact rule application has been broken and needs some fixing. 
	We call this a \emph{broken match for a consistency pattern} or, short, a \emph{broken consistency match}.
	Practically, we will exploit an incremental pattern matcher to notify us about such disappearances. 
	\begin{figure}[ht]
		\centering
		\includegraphics[width=1.0\columnwidth]{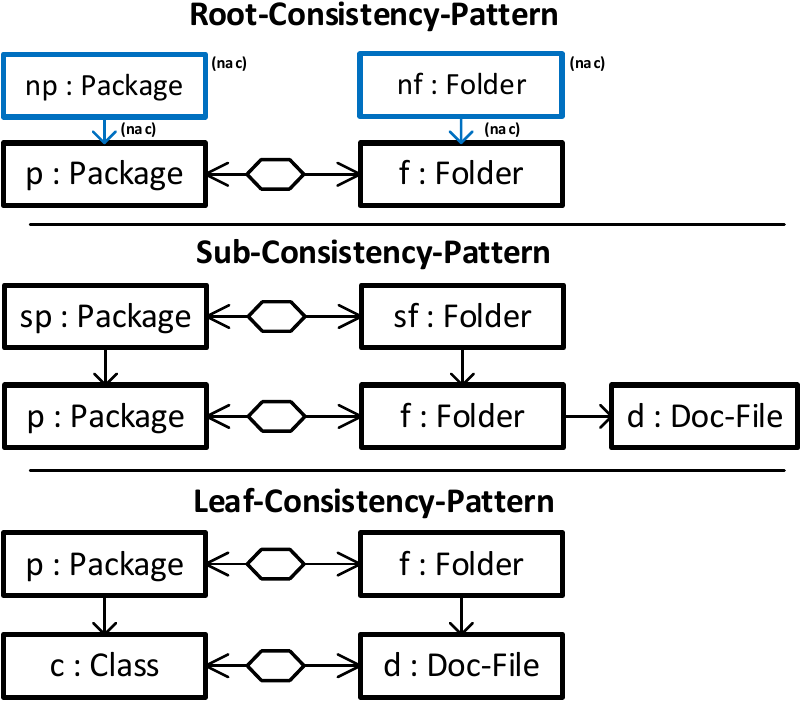}
		\caption{Example: Consistency Patterns}
		\label{fig:consistency}
	\end{figure}
\end{example}

\subsection{Sequential independence}
The proof of correctness of our synchronization approach relies on the notion of \emph{sequential independence}. 
Transformations that are sequentially independent can be performed in arbitrary order.

\begin{definition}[Sequential independence]
	Given two transformation steps $G \Rightarrow_{r_1,m_1} H_1 \Rightarrow_{r_2,m_2} X$, via plain rules $r_1,r_2$ these are \emph{sequentially independent} if 
	\begin{equation}\label{eq:seq-independence}
		n_1(R_1) \cap m_2(L_2) \subseteq n_1(K_1) \cap m_2(K_2)
	\end{equation}
	where $n_1$ is the comatch of the first transformation.
\end{definition}

By the Local Church-Rosser Theorem the order of sequentially independent transformation can be switched. This means that, given a sequentially independent transformation sequence $G \Rightarrow_{r_1,m_1} H_1 \Rightarrow_{r_2,m_2} X$, there exists a sequentially independent transformation sequence $G \Rightarrow_{r_2,m_2^\prime} H_2 \Rightarrow_{r_1,m_1^\prime} X$~\cite[Theorem~3.20]{EEPT06}. 
	If $r_1$ and $r_2$ are equipped with NACs $\mathit{NAC}_1$ and $\mathit{NAC}_2$, respectively, transformation steps as above are \emph{sequentially independent} if condition~(\ref{eq:seq-independence}) holds and moreover, the thereby induced matches $m_2^\prime: L_2 \to G$ and $m_1^\prime: L_1 \to H_2$ both satisfy the respective sets of NACs. In particular, the Local Church-Rosser Theorem still holds. 
	
	In our setting of graph transformation, it is easy to check the sequential independence of transformations~\cite{EEPT06,EGHLO14}. 
	A sequence $t_1;t_2$ of two transformation steps is sequentially independent if and only if the following holds.
	\begin{itemize}
    \item $t_2$ does not match an element that $t_1$ created.
		\item$t_2$ does not delete an element that $t_1$ matches.
		\item  $t_2$ does not create an element that $t_1$ forbids. 
		\item $t_1$ does not delete an element that $t_2$ forbids. 
	\end{itemize}

\section{Short-cut Rules}
\label{sec:sc-rules}
Short-cut rules were introduced in \cite{FKST18} to take back an application of a TGG rule and to apply another one instead. This exchange of application shall be performed such that information loss is avoided. This means that model elements are check for reuse before deleting them.  We recall the construction of short-cut rules first and discuss their expressivity thereafter. Finally, we identify conditions for language-preserving applications of short-cut rules.

\subsection{Construction of short-cut rules}
We recall the construction of short-cut rules in a semiformal way and reuse an example of~\cite{FKST18} for illustration; a formal treatment (in a category-theoretical setting) can be found in that paper. 
Given an inverse monotonic rule (i.e., a rule that purely deletes) and a monotonic rule, a \shortcutRule{} combines their respective actions into a single rule. 
Its construction allows to identify elements that are deleted by the first rule as recreated by the second one. 
To motivate the construction, assume two monotonic rules $r_1: L_1 \rightarrow R_1$ and $r_2: L_2 \rightarrow R_2$
be given. 
Applying the inverse rule of $r_1$ to a triple graph $G$, provides an image of $L_1$ in the resulting triple graph $H$. 
When applying $r_2$ thereafter, the chosen match for $L_2$ in $H$ may intersect with the image of $L_1$ yielding a triple graph $L_{\cap}$. 
This intersection can also be understood as saying that $L_{\cap}$ provides a partial match for $L_2$. 
The inverse application of the first rule deletes elements which may be recreated again.
In this case, it is possible to extend the sub-triple graph $L_{\cap}$ of $H$ 
to a sub-triple graph $R_{\cap}$ of $H$ with these elements. 
In particular, $R_{\cap}$ is a sub-triple graph of $R_1$ and $R_2$ as it includes elements only that have been deleted by the first rule and created by the second. 
Based on this observation, the construction of short-cut rules is defined as follows (slightly simplified and directly merged with an example):
\begin{construction}[Short-cut rule]
Let two plain, monotonic rules $r_1 = L_1 \to R_1$ and $r_2= L_2 \to R_2$ be given. 
A short-cut rule $r_{\mathit{sc}}$ for the rule pair $(r_1,r_2)$, where $r_1$ is considered to be applied inversely, is constructed in the following way:
\begin{figure*}
	\centering
	\includegraphics[width=.8\textwidth,trim={1.5mm 1.5mm 1.5mm 1.5mm},clip]{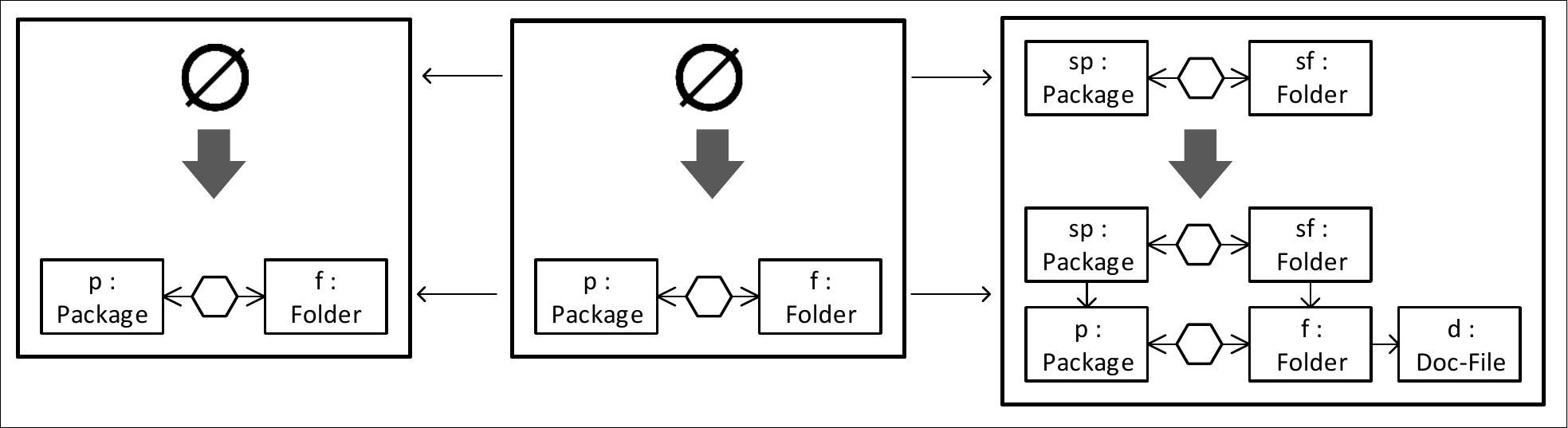}
	\caption{A common kernel rule pair (\firstTGGRule{},\secondTGGRule{}). The names of the nodes indicate their mappings and the rules are depicted top-down.}
	\label{fig:example-common-kernel}
\end{figure*}
\begin{enumerate} 
	\item \emph{Choice of common kernel:} 
	A (potentially empty) sub-triple graph $L_{\cap}$ of $L_1$ and $L_2$ and a sub-triple graph $R_{\cap}$  of $R_1$ and $R_2$ with $L_{\cap} \subseteq R_{\cap}$ are chosen. 
	We call $\Lcap \subseteq \Rcap$ a \emph{common kernel} of both rules. 

	In Fig.~\ref{fig:example-common-kernel}, an example of such a common kernel is given. 
	It is a common kernel for rule pair (\firstTGGRule{}, \secondTGGRule{}). The common kernel is depicted in the center of Fig.~\ref{fig:example-common-kernel}.
	This choice of a common kernel will lead to \firstSCRule{} as resulting short-cut rule.
	In this example, $L_{\cap}$ is empty and $R_{\cap}$ extends $L_{\cap}$ by identifying the \packages{} \p, \folders{} \f, and the correspondence node in between. 
	The elements of $R_{\cap} \setminus L_{\cap}$, called \emph{recovered elements}, are to become the elements that are preserved by an application of the short-cut rule compared to reversely applying the first rule followed by applying the second one (provided that these applications overlap in $L_{\cap}$). 
	In the example case, the whole graph $R_{\cap}$ is recovered as $L_{\cap}$ is empty.

	\item \emph{Construction of LHS and RHS:}
	One first computes the union $L_{\cup}$ of $L_1$ and $L_2$ along $L_{\cap}$. 
	The result is then united with $R_1$ along $L_1$ and $R_2$ along $L_2$, respectively, to compute the LHS and the RHS of the short-cut rule. 
	Figure~\ref{fig:example-construction-LHS-RHS} displays this. 
	
	\item \emph{Interface construction:}
	The interface $K$ of the short-cut rule is computed by taking the union of $L_{\cup}$ and $R_{\cap}$ along $L_{\cap}$. 
	For our example, this construction is depicted in Figure~\ref{fig:example-construction-interface}.
	The elements of $L_2 \setminus L_{\cap}$ are called \emph{presumed elements} since, given a match for the inverse first rule, i.e., for $R_1$, these are exactly the elements needed to extend this match to a match of the short-cut rule. 
	In our example, these are the \package{} \ssp{}, the \folder{} \ssf{}, and the correspondence node in between.
\end{enumerate}

\begin{figure*}
	\centering
	\includegraphics[width=\textwidth,trim={1.5mm 1.5mm 1.5mm 1.5mm},clip]{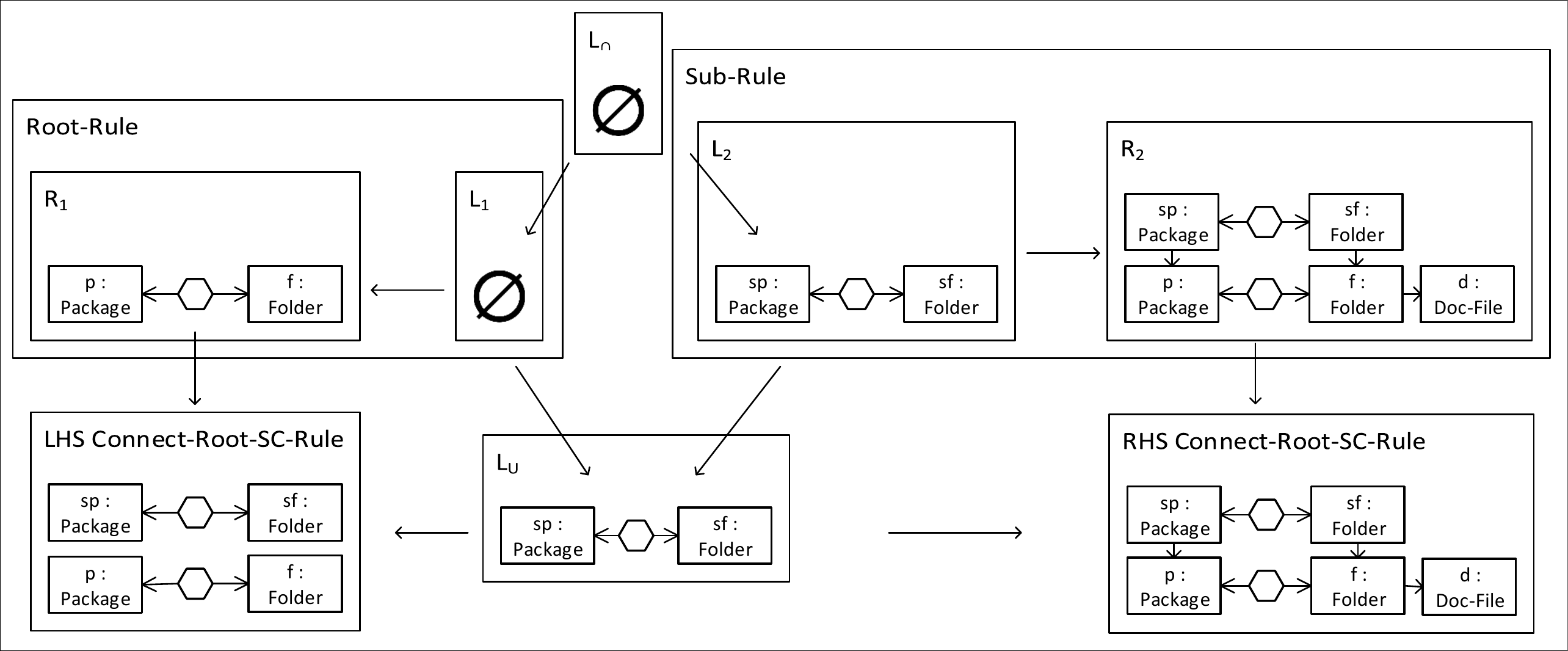}
	\caption{Constructing the LHS and the RHS of the short-cut rule \firstSCRule{}}
	\label{fig:example-construction-LHS-RHS}
\end{figure*}
\end{construction}

\begin{figure}
	\centering
	\includegraphics[width=\columnwidth,trim={1.5mm 1.5mm 1.5mm 1.5mm},clip]{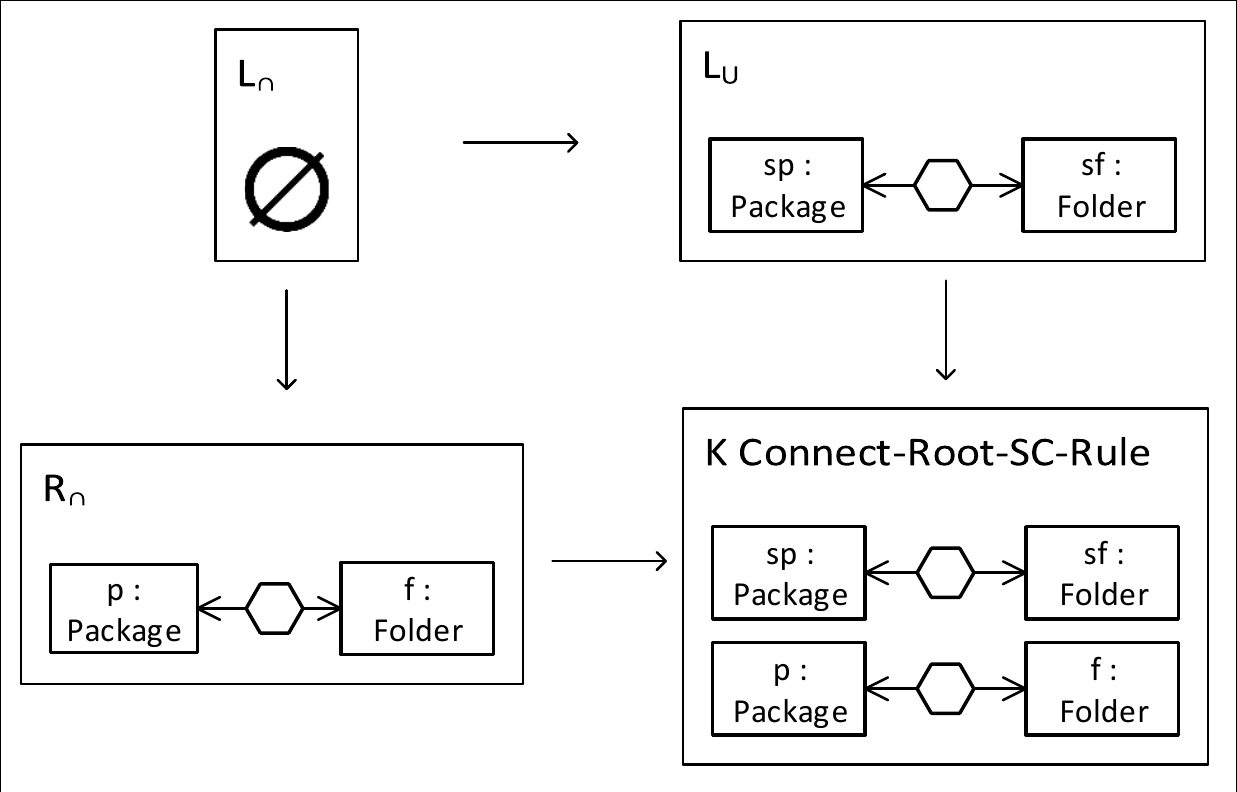}
	\caption{Constructing the interface of the short-cut \firstSCRule{}; the interface is the resulting graph in the bottom right corner.}
	\label{fig:example-construction-interface}
\end{figure}

\begin{example}
	More examples of short-cut rules are depicted in Fig.~\ref{fig:scRules}. 
	Both, \firstSCRule{} and \secondSCRule{}, are constructed for the rules \firstTGGRule{} and \secondTGGRule{}. Switching the role of of the inverse rule, two short-cut rules can be constructed having equal common kernels. 
	In both cases, the \package{} \p{}, the \folder{} \f{} and the correspondence node between them are recovered elements, as these elements would have been deleted and re-created otherwise. 
	While in \firstSCRule{}, the presumed elements are the \package{} \ssp{} and the \folder{} \ssf{} with a correspondence node in between, the set of presumed elements of \secondSCRule{} is empty.
	
	Another possible common kernel for \firstTGGRule{} and \secondTGGRule{} is one where $R_{\cap}$ is an empty triple graph as well. 
As the resulting short-cut rule just copies both rules (one of them inversely) next to each other, this rule is not interesting for our desired application. 
\end{example}

\subsection{Expressivity of short-cut rules}
Given a set of rules, there are {\em two degrees of freedom} when deciding which short-cut rules to derive from them: 
First, one has to choose for which \emph{pairs of rules} short-cut rules shall be derived. 
Secondly, given a pair of rules, there is typically not only one way to construct a short-cut rule for them:  
In general, there are different choices for a \emph{common kernel}. 
However, when fixing a common kernel, i.e., $\Lcap$ and $\Rcap$, the result of the construction is uniquely determined. If, moreover, the LHSs and RHSs of the rules are finite, the set of possible common kernels is finite as well.

As short-cut rules correspond to possible (complex) edits of a triple graph, the more short-cut rules are derived, the more user edits are available which can directly be propagated by the corresponding repair rules. 
But the number of rules that has to be computed (and maintained throughout the synchronization process) in this way, would quickly grow. 
And maybe several of the constructed rules would capture edits that are possible in principle but unlikely to ever be performed in a realistic scenario. 
Hence, some trade-off between expressivity and maintainability has to be found.

We shortly discuss these effects of choices: 
The construction of short-cut rules is defined for \emph{any two} monotonic rules~\cite{FKST18} -- we do not need to restrict to the rules of a given TGG but may also use monotonic rules that have been constructed as so-called \emph{concurrent rules}~\cite{EEPT06} of given TGG rules as input for the short-cut rule construction. 
A concurrent rule combines the actions of two (or more) subsequent rule applications into a single rule. 
Hence, deriving short-cut rules from concurrent rules that have been built of given TGG rules leads to short-cut rules that capture even more complex edits into a single rule. 
The next example presents such a derived short-cut rule. 
While our conceptual approach is easily extended to support such rules, we currently stick with short-cut rules directly derived from a pair of rules of the given TGG in our implementation.

\begin{example}
	The short-cut rule \fourthSCRule{} depicted in Fig.~\ref{fig:sc-rule-4} is not directly derived of the TGG rules depicted in Fig.~\ref{fig:tggRules}. 
	Instead, the concurrent rule of two given applications of \secondTGGRule{} is constructed first. 
	This concurrent rule directly creates a chain of two \packages{} and \folders{} into an existing pair of \package{} and \folder{}. 
	The rule in  Fig.~\ref{fig:sc-rule-4} is a short-cut rule of this concurrent rule and \secondTGGRule{}. It takes back the creation of a chain such that the bottom package is directly included in the top package in Fig.~\ref{fig:sc-rule-4}.
	
	\begin{figure}
		\includegraphics[width=\columnwidth]{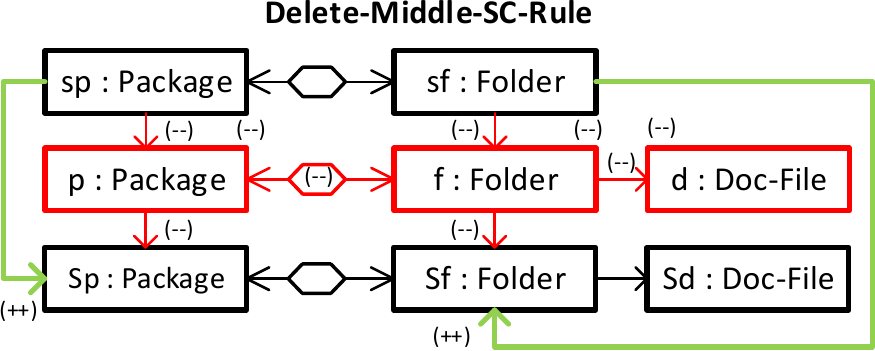}
		\caption{Example for a short-cut rule not directly derived from the rules of our example TGG}
		\label{fig:sc-rule-4}
	\end{figure}
\end{example}

Concerning the choice of a common kernel, we follow two strategies.
In both strategies, we overlap as many of the newly created elements of the two input rules as possible since these are the elements that we try to preserve. 

A \emph{minimal overlap} overlaps created elements only, i.e. no context elements. 
An example is \secondTGGRule{}, which overlapped with itself, results in \thirdSCRule{} and which corresponds to a move refactoring step. 

A \emph{maximal overlap} overlaps not only created elements of both rules but also context elements.
Creating such an overlap for \secondTGGRule{} with itself would result in the \emph{Sub-Consistency-Pattern}, which has no effect when applied.
However, when overlapping different rules with each other, it is often useful to re-use context elements. 
This is the case, for example, for \emph{VariableDec-2-Parameter-Rule} and \emph{TypeAccess-2-ReturnType-Rule} of our evaluation rule set in Fig.~\ref{fig:eval_rule_set_full} below.
A full overlap between both rules would allow to transform a signature parameter to a return parameter of the same method and of the same type and, vice versa.

Both strategies aim to create different kinds of short-cut rules with specific purposes.
Since generating all possible overlaps and thus short-cut rules is expensive, we chose a heuristic approach to generate a useful subset of them.

As we are dealing with triple graphs being composed of source, target and correspondence graphs, the overlap of source graphs should correspond to that of target graphs. 
This restricts the kind of \enquote{reuse} of elements the derived short-cut rules enable. 
The allowance of any kind of overlap may include unintended ones. 
We argue for the usefulness of these strategies in our evaluation in Sect.~\ref{sec:implAndEvaluation}. 

\subsection{Language preserving short-cut rule applications} 
The central intuition behind the construction of short-cut rules is to replace the application of a monotonic triple rule by another one. 
In this sense, a short-cut rule captures a complex edit operation on triple graphs that (in general) cannot be performed directly using the rules of a TGG. 
We illustrate this behaviour in the following. Subsequently, we discuss the circumstances under which applications of short-cut rules are \enquote{legal} in the sense that the result still belongs to the language of the respective TGG. 

\medskip

Let a TGG $\mathit{GG}$ and a sequence of transformations
\begin{equation}\label{eq:original-sequence}
	G_0 \Rightarrow_{r_1,m_1} G_1 \Rightarrow_{r_2,m_2} G_2 \Rightarrow \dots \Rightarrow_{r_t,m_t} G_t
\end{equation}
be given where all the $r_i$, $1 \leq i \leq t$, are rules of $\mathit{GG}$, all the $m_i$ denote the respective matches, and $G_0 \in \mathcal{L}(\mathit{GG})$; in particular $G_t \in \mathcal{L}(\mathit{GG})$ as well.
Fixing some $j \in \{1, \dots, t\}$ and some rule $r$ of $\mathit{GG}$, we construct a short-cut rule $r_{sc}$ for $(r_j, r)$ with some common kernel $\Lcap \subseteq \Rcap$. 
Next, we can consider the transformation sequence
\begin{equation*}
	G_0 \Rightarrow_{r_1,m_1} G_1 \Rightarrow_{r_2,m_2} G_2 \Rightarrow \dots \Rightarrow_{r_t,m_t} G_t \Rightarrow_{r_{sc},m_{sc}} G_t'
\end{equation*}
that arises by appending an application of $r_{sc}$ to transformation sequence~(\ref{eq:original-sequence}). 
Under certain technical cir\-cumstances (which we will state below) this transformation sequence is equivalent\footnote{The formal notion of equivalence used here is called \emph{switch equivalence} and captures the idea that, in case of sequential independence, the order of rule applications might be switched while using basically the same match for each rule application and receiving the same result; compare, e.g.,~\cite{Kreowski86,BCHKS14}.} to the sequence 
\begin{equation}\label{eq:final-sequence}
\begin{split}
	G_0 \Rightarrow_{r_1,m_1} G_1 \Rightarrow \dots \Rightarrow_{r_{j-1},m_{j-1}} G_{j-1} \Rightarrow_{r,m_{sc}^{\prime}} G_j^{\prime} \\
\Rightarrow_{r_{j+1},m_{j+1}^{\prime}} \dots \Rightarrow_{r_t,m_t^{\prime}} G_t^{\prime\prime} 
\end{split}
\end{equation}
where the application of $r_j$ at match $m_j$ is replaced by an application of $r$ at a match $m_{sc}\rq{}$ that is derived from the match $m_{sc}$ of the \shortcutRule{}. The following matches $m_{j+1},\allowbreak \dots,\allowbreak m_{t}$ have been adapted accordingly. 
They still match the same elements but formally they do so in other triple graphs. 
In particular, $G_t^{\prime\prime}$, the result of the transformation sequence (\ref{eq:final-sequence}), \emph{is isomorphic} to $G_t^{\prime}$ and hence, $G_t^{\prime}$ can be understood as arising by replacing the $j$-th rule application in the transformation sequence (\ref{eq:original-sequence}) by an application of the rule $r$; thus, $G_t^{\prime}$ also belongs to the language of the TGG: 
The sequence~(\ref{eq:final-sequence}) starts at a triple graph $G_0 \in \mathcal{L}(\mathit{GG})$ and solely consists of applications of rules from $\mathit{GG}$. 

\begin{figure*}
	\centering
	\includegraphics[width=.75\textwidth]{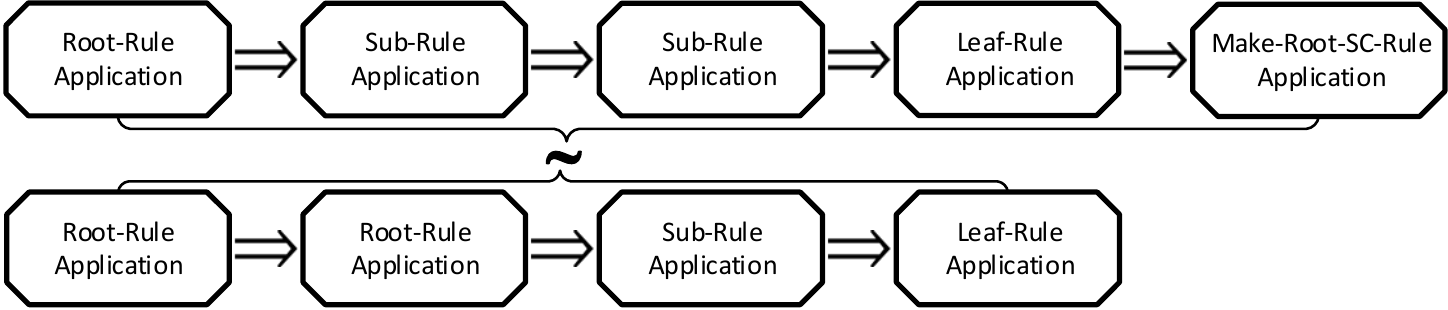}%
	\caption{Example: Transforming sequences of rule applications by applying short-cut rules}%
	\label{fig:ruleApplications_Solved}%
\end{figure*}

\begin{example}
	Consider the triple graph depicted in Fig.~\ref{fig:translationExample} (a). 
	It arises by applying \firstTGGRule{}, followed by two applications of \secondTGGRule{}, and finally an application of \thirdTGGRule{}. 
	When matched as already described in the introductory example, an additional application of \secondSCRule{} to this triple graph results in the one depicted in Fig.~\ref{fig:translationExample} (c). 
	Alternatively, this can be derived by two applications of \firstTGGRule{}, followed by an application of \secondTGGRule{} and \thirdTGGRule{} each. 
	As schematically depicted in Fig.~\ref{fig:ruleApplications_Solved}, the application of the \shortcutRule{} \secondSCRule{} transforms an transformation sequence deriving the first triple graph into a transformation sequence deriving the second one by replacing an application of \secondTGGRule{} by one of \firstTGGRule{}. 
\end{example}
 
In the following, we state when the above described behaviour is the case (in a somewhat less technical language than originally used).

\begin{theorem}[{\cite[Theorem~8]{FKST19}}]\label{thm:valid-applications}
Let the transformation sequence~(\ref{eq:original-sequence}) be given and let $r_{sc}$ be a \shortcutRule{} that is derived from $(r_j, r)$. 
If the following three conditions are met, this sequence is equivalent to sequence~(\ref{eq:final-sequence}) where original TGG rules are applied only.  
\begin{enumerate}
	\item \emph{Reversing match:} The application of $r_{sc}$ at $m_{sc}$ reverses the application of $r_j$, i.e., $n_j(R_j) = m_{\mathit{sc}}|_{R_j}(R_j)$. 
	
	\item \emph{Sequential independence:}
	\begin{enumerate}
		\item \emph{Non-disabling match:} The application of $r_{sc}$ at $m_{sc}\rq{}$ does not delete elements used in the applications of $r_{j+1}, \dots, r_t$. 
	
		\item \emph{Context-preserving match:} The match $m_{sc}$ for $r_{sc}$ already exists in $G_{j-1}$. 
		Since the assumption on the match to be reversing already ensures this for elements of $L_{\mathit{sc}}$ that stem from $R_j$, context-preservation ensures in particular that the \emph{presumed elements} of $r_{sc}$ are matched to elements already existing in $G_{j-1}$. 
	\end{enumerate}
\end{enumerate}
\end{theorem}

\begin{example}
	We illustrate each of the above mentioned conditions:
	\begin{enumerate}
		\item \emph{Reversing match:} In our example of matching \firstSCRule{} to the triple graph (c) in Fig.~\ref{fig:translationExample} this means that its nodes \p{} and \f{} (and the correspondence node in between) are allowed to be matched to elements only that have been created using \firstTGGRule{}. 
		In this way, it is avoided to mis\-use the rule to introduce \packages{} (and \folders{}) that are contained by more than one \package{} (or \folder{}). 
		
		\item \emph{Non-disabling match:} For example, \fourthSCRule{} from Fig.~\ref{fig:sc-rule-4} is not allowed to delete \packages{} and \folders{} that already contain \classes{} or \docs{}, respectively. 
		
		\item \emph{Context preserving match:} Returning to our example of matching \firstSCRule{} to the triple graph (c) in Fig.~\ref{fig:translationExample} this means that as soon as nodes \subP{} and \subF{} in that triple graph have been chosen as matches for the nodes \p{} and \f{} of \firstSCRule{}, the nodes \leafP{} and \leafF{} are not allowed to be chosen as matches for nodes \ssp{} and \ssf{} of \firstSCRule{}. 
		The creation of \leafP{} and \leafF{} depends on \subP{} and \subF{} being created first. 
		In this way, the introduction of cyclic dependencies between elements is avoided.
	\end{enumerate}
\end{example}

\section{Constructing Language-Preserving Repair Rules}
\label{sec:constructing-repair-rules}
In this section, we formally define the derivation of repair rules from a given TGG and characterize valid applications of these.
Our general idea is to construct \emph{repair rules} that can be used during model synchronization processes that are based on the formalism of TGGs. 
Our construction of such repair rules is based on \emph{\shortcutRules{}} which we recalled in Section~\ref{sec:sc-rules}.  

\subsection{Deriving repair rules from short-cut rules}
Having defined short-cut rules, they can be operationalized to get edit rules for source graphs and forward rules that repair these edits. As such edits may delete source elements, correspondence elements may be left without corresponding source elements. Hence, the resulting triple graphs show a form of partiality. 
They are called \emph{partial triple graphs}. 
Given a model, formally considered as triple graph $G_S \xleftarrow{\sigma_G} G_C \xrightarrow{\tau_G} G_T$, a user edit on $G_S$ may consist of the deletion and/or creation of graph elements, resulting in a graph $G_S^{\prime}$. 
In general, the \enquote{old} correspondence morphism $\sigma_G: G_C \to G_S$ does not extend to a correspondence morphism from $G_C$ to $G_S^{\prime}$: 
The user might have deleted elements in the image of $\sigma_G$. 
However, there is 
a \emph{partial morphism} $\sigma_G^{\prime}: G_C \dashrightarrow G_S^{\prime}$ that is defined for all elements whose image under $\sigma_G$ still exists.
\begin{definition}[Partial triple graph]
	A \emph{partial graph morphism} $f: A \dashrightarrow B$ is a graph morphism $f: A^{\prime} \to B$ where $A^{\prime}$ is a subgraph of $A$; $A^{\prime}$ is called the \emph{domain} of $f$.
	
	A \emph{partial triple graph} $G^{\prime} = G_S^{\prime} \xdashleftarrow{\sigma_G^{\prime}} G_C^{\prime} \xdashrightarrow{\tau_G^{\prime}} G_T^{\prime}$ consists of three graphs $G_S^{\prime},G_C^{\prime},G_T^{\prime}$ and two partial graph morphisms $\sigma_G^{\prime}: G_C^{\prime} \dasharrow G_S^{\prime}$ and $\tau_G^{\prime}: G_C^{\prime} \dasharrow G_T^{\prime}$. 
	
	Given a triple graph $G = (G_S \xleftarrow{\sigma_G} G_C \xrightarrow{\tau_G} G_T)$ and a user edit of $G_S$ that results in a graph $G_S^{\prime}$, the \emph{partial triple graph induced by the edit} is $G_S^{\prime} \xdashleftarrow{\sigma_G^{\prime}} G_C \xrightarrow{\tau_G} G_T$ where $\sigma_G^{\prime}$ is obtained by restricting $\sigma_G$ to those elements $x$ of $G_C$ (node or edge) for which $\sigma_G(x) \in G_S$ is still an element of $G_S^{\prime}$.
\end{definition}
According to the above definition, triple graphs are special partial triple graphs, namely those, where the domain of both partial correspondence morphisms is the whole correspondence graph $G_C$. 

When operationalizing short-cut rules, i.e., splitting them into a \emph{source} and a \emph{forward rule}, we also have to deal with this kind of partiality: 
In contrast to the rules of a given TGG, a short-cut rule might delete an element. 
Hence, its forward rule might need to contain a correspondence element for which the corresponding source element is missing; it is referenced in the short-cut rule. 
This element is deleted by the corresponding source rule. 

\begin{definition}[Source and forward rule of short-cut rule. Repair rule]
	Given a pair $(r_1,r_2)$ of plain, monotonic triple rules with short-cut rule $r_{\mathit{sc}} = (L_{\mathit{sc}} \xleftarrow{l_{\mathit{sc}}} K_{\mathit{sc}} \xrightarrow{r_{\mathit{sc}}}  R_{\mathit{sc}})$, the \emph{source} and \emph{forward rule} of $r_{\mathit{sc}}$ are defined as 
	\begin{equation*}
		r_{\mathit{sc}}^S \coloneqq (L_{\mathit{sc}}^S \xleftarrow{(l_{\mathit{sc},S}, \mathit{id}_{\emptyset}, \mathit{id}_{\emptyset})} K_{\mathit{sc}}^S \xrightarrow{(r_{\mathit{sc},S}, \mathit{id}_{\emptyset}, \mathit{id}_{\emptyset})} R_{\mathit{sc}}^S)
	\end{equation*}
	and 
	\begin{equation*}
		r_{\mathit{sc}}^F \coloneqq (L_{\mathit{sc}}^F \xleftarrow{(id_{R_{\mathit{sc},S}}, l_{\mathit{sc},C}, l_{\mathit{sc},T})} K_{\mathit{sc}}^F \xrightarrow{(id_{R_{\mathit{sc},S}}, r_{\mathit{sc},C}, r_{\mathit{sc},T})} R_{\mathit{sc}}^F)
	\end{equation*}
	where
	\begin{align*}
		L_{\mathit{sc}}^S	& \coloneqq (L_{\mathit{sc},S} \leftarrow \emptyset \rightarrow \emptyset), \\
		K_{\mathit{sc}}^S	& \coloneqq (K_{\mathit{sc},S} \leftarrow \emptyset \rightarrow \emptyset), \\
		R_{\mathit{sc}}^S	& \coloneqq (R_{\mathit{sc},S} \leftarrow \emptyset \rightarrow \emptyset), \\
		L_{\mathit{sc}}^F	& \coloneqq (R_{\mathit{sc},S} \dashleftarrow L_{\mathit{sc},C} \rightarrow L_{\mathit{sc},T}), \\
		K_{\mathit{sc}}^F & \coloneqq (R_{\mathit{sc},S} \leftarrow K_{\mathit{sc},C} \rightarrow K_{\mathit{sc},T}), \text{ and } \\
		R_{\mathit{sc}}^F & \coloneqq (R_{\mathit{sc},S} \leftarrow R_{\mathit{sc},C} \rightarrow R_{\mathit{sc},T}) \enspace .
	\end{align*}
	
	Given a TGG $\mathit{GG}$, a \emph{repair rule for $\mathit{GG}$} is the forward rule $r_{\mathit{sc}}^F$ of a short-cut rule $r_{\mathit{sc}}$ where $r_{\mathit{sc}}$ has been constructed from a pair of rules of $\mathit{GG}$.
\end{definition}

For more details (in particular, the definition of morphisms between partial triple graphs), we refer the interested reader to the literature~\cite{FKST19,KFST19}.  
In this paper, we are more interested in conveying the intuition behind these rules by presenting examples.
We next recall the most important property of this operationalization, namely that, as in the monotonic case, an application of a short-cut rule corresponds to the application of its source rule, followed by an application of the forward rule if consistently matched.

\begin{theorem}[{\cite[Theorem~7]{FKST19} and \cite[Theorem~23]{KFST19}}]\label{thm:equivalence}
	Given a short-cut rule $r_{\mathit{sc}}$, there is a transformation 
	\begin{equation*}
		(G_S \leftarrow G_C \rightarrow G_T) \Rightarrow_{r_{\mathit{sc}},m_{\mathit{sc}}} (H_S \leftarrow H_C \rightarrow H_T)
	\end{equation*}
	via this short-cut rule if and only if there is a transformation
	\begin{align*}
		(G_S \leftarrow G_C \rightarrow G_T) 	& \Rightarrow_{r_{\mathit{sc}}^S,m_{\mathit{sc}}^S} (H_S \dashleftarrow G_C \rightarrow G_T) \\ 
																					& \Rightarrow_{r_{\mathit{sc}}^F,m_{\mathit{sc}}^F} (H_S \leftarrow H_C \rightarrow H_T)
	\end{align*}
applying source rule $r_{\mathit{sc}}^S$ with match $m_{\mathit{sc}}^S = (m_{\mathit{sc},S},\emptyset,\emptyset)$ and forward rule $r_{\mathit{sc}}^F$ at match $m_{\mathit{sc}}^F  = (n_{\mathit{sc},S},m_{\mathit{sc},C},m_{\mathit{sc},T})$. 
\end{theorem}

For practical applications, repair rules should also be equipped with filter NACs. 
Let the repair rule $r_{\mathit{sc}}^F$ be obtained from a short-cut rule $r_{\mathit{sc}}$ that has been computed from rule pair $(r_1, r_2)$, both coming from a given TGG. 
As the application of $r_{\mathit{sc}}^F$ replaces an application of $r_1^F$ by one of $r_2^F$, $r_{\mathit{sc}}^F$ should be equipped with the filter NAC of $r_2^F$. 
However, just copying that filter NAC would not preserve its semantics; a more refined procedure is needed.
The LHS of $r_2^F$ is a subgraph of the one of $r_{\mathit{sc}}^F$ by construction. 
There is a known procedure, called \emph{shift along a morphism}, that \enquote{moves} an application condition from a subgraph to the supergraph 
preserving its semantics~\cite[Lemma~3.11 and Construction~3.12]{EGHLO14}. 
We use this construction to compute the filter NACs of repair rules. 
By using this known construction, the filter NACs we construct for our repair rules have the following property:

\begin{lemma}[{\cite[Lemma~3.11 and Construction~3.12]{EGHLO14}.}]\label{lem:property-filter-NACs}
	Let $r_{\mathit{sc}}$ be a plain short-cut rule obtained from the pair of monotonic rules $(r_1,r_2)$ where the forward rule $r_2^F$ is equipped with a set $\mathit{NAC}_2^F$ of filter NACs. 
	Let $\mathit{NAC}_{\mathit{sc}}^F$ be the set of NACs computed by applying the shift construction to $\mathit{NAC}_2^F$ along the inclusion morphism $\iota: L_2^F \hookrightarrow L_{\mathit{sc}}^F$ of the LHS of $r_2^F$ into the LHS of $r_{\mathit{sc}}$ (which exists by construction). 
	
	Then, an injective match $m_{\mathit{sc}}^F$ for $r_{\mathit{sc}}^F$ (into any partial triple graph $G$) satisfies the set of NACs $\mathit{NAC}_{\mathit{sc}}^F$ if and only if the induced injective match $m_{\mathit{sc}}^F \circ \iota$ for $r_2^F$ satisfies $\mathit{NAC}_2^F$.
\end{lemma}

\begin{example}
	The forward rules of the short-cut rules in Fig.~\ref{fig:scRules} are depicted in Fig.~\ref{fig:fwdSCRules}. 
	\secondRRule{} is derived to replace an application of \secondTGGForwardRule{} by one of \firstTGGForwardRule{}. 
	This forward rule is equipped with a filter NAC which ensures that the rule is used only to translate \packages{} at the top of a hierarchy. 
	Just copying this NAC to the \package{} \p{} in \secondRRule{} would not preserve this behaviour: 
	The rule would be applicable in situations where the \package{} to which \ssp{} is matched contains a \package{} to which \p{} is matched. 
	Shifting the NAC from \firstTGGForwardRule{} to \secondRRule{} instead, the forbidden edge between the two \packages{} is introduced in addition. 
	It ensures that \p{} can be matched to \packages{} at the top of a hierarchy, only. 

	\fourthRRule{} (see Figure~\ref{fig:fourth-repair-rule}) assumes two connected \packages{} and deletes a \folder{} between their corresponding \folders{} as well as the \doc{} contained in the deleted \folder{} and the correspondence node referencing it. 
	The LHS of this rule is a proper partial triple graph as there is a correspondence node which is not mapped to any element of the source part. 
	
	\begin{figure}%
		\includegraphics[width=\columnwidth]{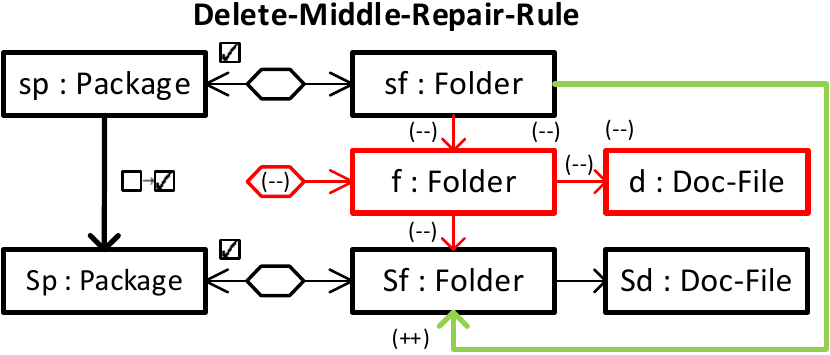}%
		\caption{Repair rule derived from \fourthSCRule{}}%
		\label{fig:fourth-repair-rule}%
	\end{figure}
\end{example}

\subsection{Conditions for valid repair rule applications}

Now, we transfer the results obtained so far to the case of repair rules. 
To do so, we first define \emph{valid matches} for repair rules (in a restricted kind of transformation sequences). 

\begin{definition}[Valid match for repair rule]\label{def:valid-repair-match}
	Let a TGG $\mathit{GG}$ and a consistently-marking transformation sequence 
	\begin{equation}\label{eq:orig-trafo}
		G_0 \Rightarrow_{r_1^{\mathit{FN}},m_1^F} G_1 \Rightarrow_{r_2^{\mathit{FN}},m_2^F} \dots \Rightarrow_{r_t^{\mathit{FN}},m_t^F} G_t
	\end{equation}
	via forward rules $r_i^{\mathit{FN}},\ 1 \leq i \leq t$, (possibly with filter NACs) of $\mathit{GG}$ be given. 
	Let 
	\begin{equation*}
		G_i = (G_{0,S} \leftarrow G_{i,C} \rightarrow G_{i,T}) \enspace .
	\end{equation*}
	Let there be some source edit step
	\begin{equation*}
		G_t \Rightarrow_{r_{\mathit{sc}}^S,m_{\mathit{sc}}^S} G^\prime 
	\end{equation*}
	where $G^\prime = (H_S \dashleftarrow G_{t,C} \rightarrow G_{t,T})$, $r_{\mathit{sc}}$ is a source rule of a short-cut rule derived from a rule pair $(r_j,r)$ where $1 \leq j \leq t$ and $r$ stems from $\mathit{GG}$, and $m_{\mathit{sc}}^S|_{R_{j,S}} = n_{j,S}$, i.e., when restricted to the source part of the RHS $R_j$ of $r_j$ match $m_{\mathit{sc}}^S$ coincides with the source part of the comatch $n_j$. 
	Moreover, the application of this source edit shall not introduce a violation of any of the filter NACs of $r_1^{\mathit{FN}}, \dots, r_{j-1}^{\mathit{FN}}$.
	
	Then, a match $m_{\mathit{sc}}^F$ for the corresponding forward rule $r_{\mathit{sc}}^F$ in $G^\prime$ is \emph{valid} if the following properties hold.
	
	\begin{enumerate}
		\item \emph{Reversing match:} Given comatch $(n_{\mathit{sc},S}^S,\emptyset,\emptyset)$ of the application of the source rule $r_{\mathit{sc}}^S$, its match is 
		\begin{equation*}
			m_{\mathit{sc}}^F = (n_{\mathit{sc},S}^S,m_{\mathit{sc,C}}^F,m_{\mathit{sc,T}}^F) 
		\end{equation*}
		and also $m_{\mathit{sc,C}}^F$ and $m_{\mathit{sc,T}}^F$ coincide with $n_{j,C}$ and $n_{j,T}$ when restricted to $R_{j,C}$ and $R_{j,T}$, respectively. 
		\item \emph{Sequential independence:}
		\begin{enumerate}
			\item \emph{Non-disabling match:} The application of $r_{\mathit{sc}}^F$ does not delete elements used in the applications of $r_{j+1}^{\mathit{FN}}, \dots, r_{t}^{\mathit{FN}}$ nor does it create elements forbidden by one of the filter NACs of those forward rules. 
			\item \emph{Context-preserving match:} The presumed elements of the repair rule $r_{\mathit{sc}}^F$ (which accord to the presumed elements of the short-cut rule $r_{\mathit{sc}}$) are matched to elements of $H_S$ which are marked as translated in $G_{j-1,S}$ and in $G_{t,C}$ and $G_{t,T}$ to elements which are already created in $G_{j-1,C}$ and $G_{j-1,T}$. 
This means, elements stemming from the LHS $L$ of $r$ which have not been identified with elements from $L_j$ in the short-cut rule $r_{\mathit{sc}}$, are matched to elements already translated/existing in $G_{j-1}$. 
		\end{enumerate}
		\item \emph{Creation-preserving match:} 
		All source elements that are newly created by short-cut rule $r_{\mathit{sc}}$, i.e., the source elements of $R_S \setminus L_S$ that have not been merged with an element of $R_{j,S} \setminus L_{j,S}$ during the short-cut rule construction, are matched to elements which are yet untranslated in $G_{t,S}$. 
	\end{enumerate}
\end{definition}

Note that together, conditions 2. (a) and (b) above constitute sequential independence between the applications of $r_{j+1}^{\mathit{FN}}, \dots, r_{t}^{\mathit{FN}}$ and the one of $r_{\mathit{sc}}^{\mathit{F}}$. 
Moreover, the additional requirement on the match to be creation preserving (compared to Theorem~\ref{thm:valid-applications} for short-cut rules) originates from the fact that forward rules do not create but mark source elements.

The following corollary uses Theorem~\ref{thm:equivalence} to transfer the statement of Theorem~\ref{thm:valid-applications} to repair  rules. 
\begin{corollary}\label{cor:valid-repair-rule-applications}
	Let a TGG $\mathit{GG}$ and a consistently marking transformation sequence as in~(\ref{eq:orig-trafo}), followed by an edit step exactly as in Definition~\ref{def:valid-repair-match} above be given.  
	Then, applying $r_{\mathit{sc}}^F$ at a \emph{valid match} $m_{\mathit{sc}}^F$ in $G^{\prime}$ induces a \emph{consistently marking transformation sequence} 
	\begin{equation}\label{eq:induced-trafo}
		\begin{split}
			G_0^{\prime} \Rightarrow_{r_1^{\mathit{FN}},m_1^F} G_1^{\prime} \Rightarrow_{r_2^{\mathit{FN}},m_2^F} \dots	& \Rightarrow_{r_{j-1}^{\mathit{FN}},m_{j-1}^F} G_{j-1}^{\prime} \\
																																																& \Rightarrow_{r^{\mathit{FN}},m^F} X
		\end{split}																																						
	\end{equation}
	with $G_i^{\prime} = (H_S \dashleftarrow G_{i,C} \rightarrow G_{i,T})$ for $0 \leq i \leq j-1$. 
\end{corollary}

\begin{proof}
	For a valid match $m_{\mathit{sc}}^F$ of $r_{\mathit{sc}}^F$, by its \emph{reversing property}, the conditions of Theorem~\ref{thm:equivalence} are met. 
	Hence, we obtain a sequence
	\begin{equation*}
		G_0 \Rightarrow_{r_1^{\mathit{FN}},m_1^F} \dots \Rightarrow_{r_t^{\mathit{FN}},m_t^F} G_t \Rightarrow_{r_{\mathit{sc}},m_{\mathit{sc}}} X^\prime \enspace .
	\end{equation*}
	As a consistently marking sequence of forward rules corresponds to a sequence of TGG rule applications, and the preconditions of Theorem~\ref{thm:valid-applications} are met (\enquote{exists} is exchanged by \enquote{marked} on the source component), this sequence induces a sequence 
	\begin{equation*}
		G_0 \Rightarrow_{r_1^{\mathit{FN}},m_1^F} \dots \Rightarrow_{r_{j-1}^{\mathit{FN}},m_{j-1}^F} G_{j-1} \Rightarrow_{r,m} X 
	\end{equation*}
	(where we do not care for the further applications of forward rules).
	
	Now, we can split $r$ into its source and forward rule. 
	Its source rule is sequentially independent from the other forward rule applications: 
	$r_{\mathit{sc}}^S$ does not delete anything, the rules $r_1^{\mathit{FN}}, \dots, r_{j-1}^{\mathit{FN}}$ match, and does not create a filter NAC violation by assumption and, as a consequence, $r^S$ does not.
	Hence, by the local Church-Rosser Theorem, we might equivalently switch the application of $r^S$ to the beginning of the sequence and obtain sequence~(\ref{eq:induced-trafo}), as desired. 
	Moreover, by Lemma~\ref{lem:property-filter-NACs}, the filter NAC of $r^F$ holds whenever $m_{\mathit{sc}}^F$ satisfies the filter NAC of $r_{\mathit{sc}}^F$. 
	
	Finally, as the start of the transformation sequence (up to index $j-1$) is context preserving, and by assumption 2. (b), the match $m_{\mathit{sc}}^F$ matches presumed elements of $r_{\mathit{sc}}^F$ to already translated ones (in $H_S$) or already created ones (in $G_{j-1,C}$ and $G_{j-1,T}$), this sequence is context preserving. 
	Analogously, assumption 3. ensures that it is creation-preserving: 
	No element which is already marked as translated in $G_{t,S}$ is marked a second time. 
	Hence, the whole sequence is consistently marking. \qed
	\end{proof}

\section{Synchronization Algorithm}
\label{sec:syncProcess}
In this section, we discuss our synchronization algorithm that is based on the correct application of derived \repairRules{}. 
We first present the algorithm and consider its formal properties subsequently. 
The section closes with a short example for a synchronization based on our algorithm and a discussion of extensions and support for advanced TGG features. 

\subsection{The Basic Setup}\label{sec:basic-setup-algorithm}
We assume a TGG $\mathit{GG}$ with \emph{plain, monotonic rules} to be given. 
Its language defines consistency. 
This means that a triple graph $G = (G_S \leftarrow G_C \rightarrow G_T)$ is consistent if and only if $G \in \mathcal{L}(\mathit{GG})$. 

\paragraph{The problem.} 
A consistent triple graph $G = (G_S \leftarrow G_C \rightarrow G_T) \in \mathcal{L}(\mathit{GG})$ is given; by Lemma~1 there exists a corresponding consistently and entirely marking sequence $t$ of forward rule applications. 
After editing source graph $G_S$ we get $G\rq{} = (H_S \dashleftarrow G_C \rightarrow G_T)$. 
Generally, the result $G\rq{}$ is a partial triple graph and does not belong to $\mathcal{L}(\mathit{GG})$.  
We assume that all the edits are performed by applying source rules. 
They may be derived from the original TGG rules or from short-cut rules.  
Our goal is to provide a \emph{model synchronization algorithm} that, given $G = (G_S \leftarrow G_C \rightarrow G_T) \in \mathcal{L}(\mathit{GG})$ and $G\rq{} = (H_S \dashleftarrow G_C \rightarrow G_T)$ as input, computes a triple graph $H = (H_S \leftarrow H_C \rightarrow H_T) \in \mathcal{L}(\mathit{GG})$. 
As a side condition, we want to minimize the amount of elements of $G_C$ and $G_T$ that are deleted and recreated during that synchronization. 

\paragraph{Ingredients of our algorithm.}
We provide a rule-based model synchronization algorithm leveraging an incremental pattern matcher. 
During that algorithm, rules are applied to compute a triple graph $(H_S \leftarrow H_C \rightarrow H_T) \in \mathcal{L}(\mathit{GG})$ from the (partial) triple graph $(H_S \dashleftarrow G_C \rightarrow G_T)$. 
We apply two different kinds of rules, namely
\begin{enumerate}
	\item forward rules derived from the rules of the TGG $\mathit{GG}$ and 
	\item repair rules, i.e., operationalized short-cut rules.
\end{enumerate}
Forward rules serve to propagate the addition of elements. 
The use of these rules for model synchronization is standard.  
However, the use of additional repair rules and the way in which they are employed are conceptually novel.\footnote{Note that consistency is still defined by the (plain, monotonic) rules of the given TGG; the general repair rules are derived only to improve the synchronization process.}  
The repair rules allow to directly propagate more complex user edits. 

During the synchronization process, the rules are applied reacting to notifications by an incremental pattern matcher. 
We require this pattern matcher to provide the following information: 
\begin{enumerate}
	\item The original triple graph $G = (G_S \leftarrow G_C \rightarrow G_T)$ is \emph{covered with consistency patterns}.
When considering the induced matches for forward rules, every element of $G_S$ is marked exactly once.  The dependency relation between elements required by these matches is acyclic.
	This means that the induced transformation sequence of forward rules is consistently and entirely marking. 
	Such a sequence always exists since $G \in \mathcal{L}(\mathit{GG})$; see Lemma~\ref{lem:emccp-series}. 
	\item \emph{Broken consistency matches} are reported. 
	A match for a consistency pattern in $G$ is broken in $G\rq{}$ if one of the elements it matches or creates has been deleted or if an element has been created that violates one of the filter NACs of that consistency pattern. 
	\item The incremental pattern matcher notifies about newly occurring matches for forward rules. 
	It does so in a \emph{correct} way, i.e., it only notifies about matches that lead to consistently marking transformations. 
	\item In addition, the incremental pattern matcher informs a \emph{precedence graph}. 
	This precedence graph contains information about the mutual dependencies of the elements in the partial triple graph.
	Here, an element is dependent on another one if the forward rule application marking the former matches the latter element as required. 
	We consider the transitive closure of this relation. 
\end{enumerate}

\subsection{Synchronization Process}
\label{subsec:sync-process}
Our synchronization process is depicted in Algorithm~\ref{alg:sync}. 
It applies rules to translate elements and repair rule applications. 
In that, it applies a different strategy than suggested in~\cite{LAFVS17,Leblebici18}. 
There, invalid rule applications are revoked as long as there exist any. 
Subsequently, forward rules are applied as long as possible. 
By trying to apply a suitable repair rule instead of revoking an invalid rule application, we are able to avoid deletion and recreation of elements. 
Our synchronization algorithm is defined as follows. 
Note that we present an algorithm for synchronizing in forward direction (from source to target) while synchronizing backwards is performed analogously.

The function \textit{synchronize} is called on the current partial triple graph that is to be synchronized.
In line~\ref{alg:sync:update}, \textit{updateMatches} is called on this partial triple graph. 
It returns the set of consistency matches currently broken, a set of consistency matches being still intact, and a set of forward TGG rule matches.

By calling the function \emph{isFinished} (line~\ref{alg:sync:isFinished-call}), termination criteria for the synchronization algorithm are checked. 
If the set of broken consistency matches and the set of forward TGG rule matches are both empty and all elements of the source graph are marked as translated, the synchronization algorithm terminates (line~\ref{alg:sync:terminate}). 
Yet, if both sets are empty but there are still untranslated elements in the source graph, an exception is thrown in line~\ref{alg:sync:exception}, signaling that the (partial) triple graph is in an inconsistent state.

Subsequently, function \emph{translate} is used (line~\ref{alg:sync:translate-call}) to propagate the creation of elements: 
If the set of forward TGG rule matches is non-empty (line~\ref{alg:sync:apply}), we choose one of these matches, apply the corresponding rule, and continue the synchronization process (line~\ref{alg:sync:sync1}).
This step is done prior to any repair. 
The purpose is to create the context which may be needed to make repair rules applicable.
An example for such a context creation is the insertion of a new root \package{} which has to be translated into a root \folder{} before applying \firstRRule{} thereafter (see Fig.~\ref{fig:tggFwdRules}). 

\begin{algorithm*}
	\begin{algorithmic}[1]
		\Function{synchronize}{tripleGraph}
		\State $\mathrm{(brokenCMatches, intactCMatches,fwdMatches)  \leftarrow updateMatches(tripleGraph)}$ \label{alg:sync:update}
		\\
	 	\If{isFinished(tripleGraph, fwdMatches, brokenCMatches))} \label{alg:sync:isFinished-call}
			\State \Return
		\EndIf
		\\
		\If{translate(tripleGraph, fwdMatches)} \label{alg:sync:translate-call}
			\State \Return
		\EndIf
		\\
		\State $\mathrm{(cMatch, success)  \leftarrow repair(tripleGraph, brokenCMatches)}$ \label{alg:sync:repair-call}
		\If{!success}
			\State \textbf{throw} InconsistentStateException \label{alg:sync:exception2}
		\EndIf
		\State  \Return
		\EndFunction
		\\
		\Function{isFinished}{tripleGraph, fwdMatches, brokenCMatches} \label{alg:sync:isFinished}
			\If{isEmpty(brokenCMatches) \&\& isEmpty(fwdMatches)} \label{alg:sync:termination}
			\If{allElementsTranslated(tripleGraph.source)}
			\State \Return true \label{alg:sync:terminate}
			\Else{}
			\State \textbf{throw} InconsistentStateException \label{alg:sync:exception}
			\EndIf
			\EndIf
			\State \Return false
		\EndFunction
		\\
		\Function{translate}{tripleGraph, fwdMatches} \label{alg:sync:translate}
			\If{!isEmpty(fwdMatches)} \label{alg:sync:apply}
			\State fwdMatch = chooseMatch(fwdMatches) 
			\State $\mathrm{tripleGraph \leftarrow applyRule(tripleGraph, fwdMatch, getFWDRule(fwdMatch))}$ 
			\State synchronize(tripleGraph) \label{alg:sync:sync1}
			\State \Return true
			\EndIf
			\State \Return false
		\EndFunction
		\\
		\Function{repair}{tripleGraph, brokenCMatches} 
			\State $\mathrm{cMatch \leftarrow chooseMatch(brokenCMatches)}$ \label{alg:sync:chooseMatch}
			\State $\mathrm{scRules \leftarrow getSuitableSCRules(cMatch)}$
			\State $\mathrm{scMatches \leftarrow findSCMatches(scRules, cMatch)}$
			\\
			\While{!isEmpty(scMatches)} \label{alg:sync:while}
			\State $\mathrm{scMatch \leftarrow chooseMatch(scMatches)}$
			\If{isValidMatch(scMatch)} \label{alg:sync:isValid}
			\State $\mathrm{tripleGraph \leftarrow applyRule(tripleGraph, cMatch, getSCRule(scMatch))}$ \label{alg:sync:sync4}
			\State synchronize(tripleGraph)  \label{alg:sync:sync2}
			\State \Return (cMatch, true)
			\EndIf
			\EndWhile
			\State \Return (cMatch, false)
		\EndFunction
	\end{algorithmic}
	\caption{eMoflon -- Synchronization Process} \label{alg:sync}
\end{algorithm*}

If the above cases do not apply, there must be at least one broken consistency match and the corresponding rule application has to be repaired (line~\ref{alg:sync:repair-call}):
Hence, we choose one broken consistency match (line~\ref{alg:sync:chooseMatch}) for which a set of suitable repair rules is determined. 
A broken consistency match includes information about the rule it corresponds to (e.g., the name of the rule). 
Furthermore, it includes which elements are missing or which filter NACs are violated such that the corresponding application does not exist any more. 
We calculate the set of matches of \repairRules{} (i.e., forward \shortcutRules{}) that stem from \shortcutRules{} revoking exactly the rule that corresponds to the broken consistency match. 
In particular, by knowing which elements of a broken rule application still exist in the current source graph, we can stick to those repair rules that preserve exactly the still existing elements.  

While the calculated set of unprocessed \repairRule{} matches is not empty (line~\ref{alg:sync:while}), we choose one of these matches and check whether it is \emph{valid}. 
By constructing the partial match of a \repairRule{}, we only need to ensure that none of its \emph{presumed elements} is matched in such a way that a cyclic dependency is introduced. 
This means that they must not be matched to elements that are dependent of elements to which the \emph{recovered elements} are matched. 
If a match is valid, we apply the corresponding repair rule and continue the synchronization process (line~\ref{alg:sync:sync2}).
If no such rule or valid match is available, an exception is thrown (line~\ref{alg:sync:exception2}). 

\subsection{Formal properties of the synchronization process.}
We discuss the termination, correctness, and completeness of our synchronization algorithm. 

Our algorithm terminates as long as every forward rule translates at least one element (which is a quite common condition; compare~\cite[Lemma~6.7]{HEOCDXGE15} or \cite[Theorem~3]{Leblebici18}).
\begin{theorem}
	Let a TGG $\mathit{GG}$ with plain, monotonic rules be given. 
	If every derived forward rule of $\mathit{GG}$ has at least one source marking element, our algorithm terminates for any finite input $G\rq{} = (H_S\rq{} \dashleftarrow G_C \rightarrow G_T)$. 
\end{theorem}

\begin{proof}
	The algorithm terminates -- by either throwing an exception or returning a result -- if at one point both, the set of broken consistency matches and the set of matches for forward rules are empty; compare the function \emph{isFinished} starting in line~\ref{alg:sync:isFinished}. 
	
	The algorithm is called recursively, always applying a forward rule if a match is available. 
	As every forward rule marks at least one element as translated and forward rules are only matched in such a way that source marking elements are matched to yet untranslated ones, the application of forward rules (lines~\ref{alg:sync:apply} et seq.), i.e., the recursive call of function \emph{translate}, halts after finitely many steps. 
	Moreover, an application of a forward rule never introduces a new broken consistency match: 
	As it neither creates nor deletes elements in the source graph, it cannot delete elements matched by a consistency pattern nor create elements forbidden by one. 
	This means that, as soon as the set of broken consistency matches is empty, the whole synchronization algorithm will terminate. 
	We show that at some point this set of broken consistency matches will be empty or an exception is thrown. 
	
	Whenever the algorithm is called with an empty set of matches for forward rules, broken consistency matches are considered by applying a repair rule, i.e., by calling the function \emph{repair}. 
	New matches for forward rules can result from this; as discussed above, newly appearing matches for forward rules are unproblematic. 
	However, an application of a repair rule does not introduce a new violation of any consistency match: 
	As it does not create source elements, it cannot introduce violations of filter NACs. 
	And by the condition on valid matches to be non-disabling (condition 2. (a) in Definition~\ref{def:valid-repair-match}), no elements needed by other consistency matches are deleted.
	Hence, by application of a repair rule, the number of invalid consistency matches is reduced by one and the algorithm terminates as soon as all broken consistency matches are repaired. 
	If there is a broken consistency match that cannot be repaired -- either because no suitable repair rule or no valid match is available -- an exception is thrown and the algorithm stops.  \qed
\end{proof}

\paragraph{Correctness.} 
Upon termination without exception, our algorithm is correct.
\begin{theorem}[Correctness of algorithm]\label{thm:correctness-algorithm}
	Let a TGG $\mathit{GG}$ with plain, monotonic rules, a triple graph $G = (G_S \leftarrow G_C \rightarrow G_T) \in \mathcal{L}(\mathit{GG})$, and a partial triple graph $G\rq{} = (G_S\rq{} \dashleftarrow G_C \rightarrow G_T)$ that arises by a user edit step on the source graph be given. 
	If our synchronization algorithm terminates without exception and yields $H = (H_S \leftarrow H_C \rightarrow H_T)$ as output, then $H_S = G_S\rq{}$ and $H \in \mathcal{L}(\mathit{GG})$.
\end{theorem}

\begin{proof}
	We see immediately that $H_S = G_S\rq{}$ since none of the applied rules modifies the source graph. 
	If the synchronization process terminates without exception, all elements are translated, no matches for forward rules are found, and no consistency match is broken any more. 
	This means that the collected matches of the forward rules form an entirely marking transformation sequence. 
	By Lemma~\ref{lem:emccp-series}, we have to show that this sequence is also consistently marking. 
	Then, the matches of the forward rules that correspond to the matches of the consistency patterns that the incremental pattern matcher has collected encode a transformation sequence that allows to translate the triple graph $(H_S \leftarrow \emptyset \rightarrow \emptyset)$ to a triple graph $(H_S \leftarrow H_C \rightarrow H_T) \in \mathcal{L}(\mathit{GG})$.
	We assume that the incremental pattern matcher recognizes all broken consistency matches and reports correct matches for forward rules only. 
	This means, throughout the application of forward rules, the set of all valid consistency matches remains consistently marking. 
	We have to show that this is also the case for repair rule applications. 
	If it is, upon termination without exception, there is an entirely and consistently marking sequence of forward rules which corresponds to a triple graph from $\mathit{GG}$ by Lemma~\ref{lem:emccp-series}.
	
	Whenever we apply a repair rule we are (at least locally) in the situation of Corollary~\ref{cor:valid-repair-rule-applications}: 
	There is a (maybe empty) sequence of consistently marking forward rule applications and a suitable broken consistency pattern indicates, that a user edit step applying the source rule $r_{\mathit{sc}}^S$ of a short-cut rule $r_{\mathit{sc}}$ has taken place. 
	Applying the repair rule $r_{\mathit{sc}}^F$ at a valid match amounts to replacing the application of rule $r_j^F$, whose consistency pattern was broken, by rule $r^F$ in a consistently marking way. \qed 
\end{proof}
 
We only informally discuss \emph{completeness}. 
We understand completeness as follows: for every input $G\rq{} = (H_S \dashleftarrow G_C \rightarrow G_T)$ with $H_S \in \mathcal{L}_S(\mathit{GG})$, we obtain a result $H = (H_S \leftarrow H_C \rightarrow H_T) \in \mathcal{L}(\mathit{GG})$. 
In general, the above proposed algorithm is not complete. 
We randomly apply forward rules at available matches (without using backtracking) but the choice and order of such applications can affect the result if the final sequence of forward rule applications leads to a dead-end or translates the given source graph. 
However, the algorithm is complete whenever the set of forward rules is of such a form that the order of their application does not make a difference (somewhat more formally: they meet some kind of confluence) and the user edit is of the form discussed in Sect.~\ref{sec:basic-setup-algorithm}. 
Analogous restrictions on forward rules hold for other synchronization processes that have been formally examined for completeness~\cite{HEOCDXGE15,Leblebici18}. 
Adding filter NACs to the forward rules of a TGG is a technique that can result in such a set of confluent forward rules even if the original set of forward rules is not. 
Moreover, there are static methods to test TGGs for such a behaviour~\cite{ALST14,HEOCDXGE15}; they check for sufficient but not for necessary criteria. 
If it is known that the set of forward rules of a given TGG guarantees completeness and the edit is of a suitable kind, a thrown exception during our synchronization process implies that $H_S \notin \mathcal{L}_S(\mathit{GG})$. 

\subsection{A synchronization example}
\begin{figure*}%
	\includegraphics[width=\textwidth]{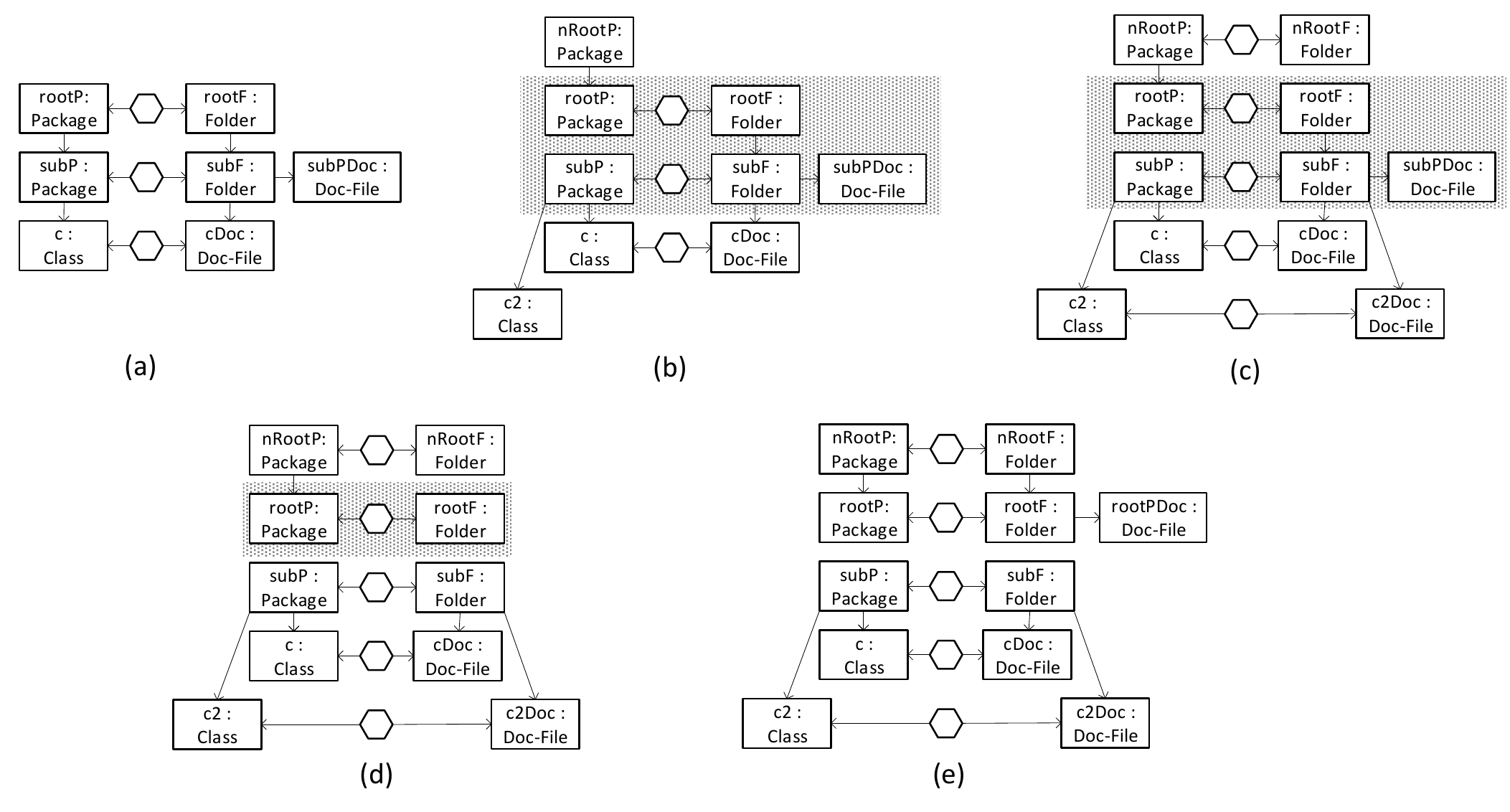}%
	\caption{Example of our proposed synchronization algorithm. Grey background indicates broken consistency matches.}%
	\label{fig:algorithm-example}%
\end{figure*}
We illustrate our synchronization algorithm with an example illustrated in Fig.~\ref{fig:algorithm-example}. 
For simplicity, we neglect the \textsf{content} attribute and concentrate on the structural behaviour. 
As a starting point, we assume that a user edits the source graph of the triple graph depicted in Fig.~\ref{fig:algorithm-example}~(a) (in the following, we will refer to the triple graphs occurring throughout the algorithm just by their numbers). 
She adds a new root package above \rootP{}, removes the link between \packages{} \rootP{} and \subP{}, and creates a further class \texttt{c2}. 
All these changes are specified by either a source rule of the TGG or the source rule of a derived short-cut rule. 
The resulting triple graph is depicted in~(b). 
The elements in front of the grey background are considered to be inconsistent, due to a broken consistency match. 
Furthermore, \texttt{c2} and \texttt{nRootP} are not translated, yet. 
In the first two passes of the algorithm, the two available matches for forward rules are applied (in random order):
\thirdTGGForwardRule{} translates the newly added \class{} \texttt{c2} and \firstTGGForwardRule{} translates the \package{} \texttt{nRootP}; this results in the triple graph (c). 
Note that the last rule application creates a match for the repair rule \firstRRule{}. 
This is the reason why we start our synchronization process with applications of forward rules. 

The incremental pattern matcher notifies about two broken consistency matches, which are dealt with in random order. 
\rootP{} is no longer a root package (which is detected by a violation of the according filter NAC in the consistency pattern) and \subP{} is now a root package (which is detected by the missing incoming edge). 
Both violations are captured by repair rules, namely \firstRRule{} and \secondRRule{}, whose applications lead to~(d) and (e). 
The algorithm terminates with a triple graph that belongs to the TGG.

\subsection{Prospect: Support of further kinds of editing and advanced TGG features}
\label{sec:prospect}
We shortly describe the support of further kinds of editing and more advanced features of TGGs by our approach to synchronization, namely attributed TGGs, rules with NACs, and support for additional attribute constraints. 

\paragraph{Further kinds of editing.}
In our implementation (see Sect.~\ref{sec:implementation}), we do not only support the addition of elements and propagation of edits that correspond to source rules of derived edit rules. 
Actually, we do not make any assumptions about the kind of editing.
This is achieved by incorporating the application of repair rules into the algorithm suggested by Leblebici et al.~\cite{LAFVS17,Leblebici18}, which has also been proved to be correct and to terminate. 
The implemented algorithm first tries to apply a forward or repair rule. 
If there is none available with a valid match, the algorithm falls back to revoking of an invalid rule application. 
This means that all elements that have been created by this rule application are deleted (and adjacent edges of deleted nodes are implicitly deleted as well).
In line with that revoking of invalid rule applications, it also allows for implicit deletion of adjacent edges in the application of  repair rules. 
In that way, the application of a repair rule might trigger new appearances of broken consistency matches. 
We are convinced that \emph{correctness} is not affected by that more general approach: 
Inspecting the proofs of Corollary~\ref{cor:valid-repair-rule-applications} and Theorem~\ref{thm:correctness-algorithm}, the key to correctness is that the sequences of currently valid consistency matches remain consistently marking. 
That is achieved via the conditions on matches for repair rules to be \emph{reversing}, \emph{context-preserving}, and \emph{creation-preserving}. 
Dropping the condition to be \emph{non-disabling} (by implicitly deleting adjacent edges) does not effect correctness, therefore. 
However, proving termination in that more general context is future work.

\paragraph{Advanced features.}
The {\em attribution of graphs} can be formalized by representing data values as special nodes and the attribution of nodes and edges as special edges connecting graph elements with these data nodes~\cite{EEPT06}. 
As the rules of a TGG are monotonic, they only set attribute values but never delete or change them. (The deletion or change of an attribute value would include the deletion of the attribution edge pointing to it.) 
The formal construction of short-cut rules is based purely on category-theoretic concepts, which can be directly applied to rules on attributed triple graphs as well. 
The properties proven for short-cut rules in~\cite{FKST18} are valid also in that case.\footnote{To be precise, in~\cite{FKST18}, all proofs are elaborated for the case of monotonic rules in an \emph{adhesive} category. Attributed triple graphs are \emph{adhesive HLR} which is a weaker notion. However, inspecting the proofs, this does not make any difference as long as the category has so-called \emph{effective pushouts}. This is known to be the case for attributed (triple) graphs; compare, e.g.,~\cite[Remark~5.57]{EEGH15}.} 
Hence, we can freely apply the construction of short-cut rules and derivation of repair rules to attributed TGGs. 
In fact, our implementation already supports attribution. 
For the propagation of attribute changes (made by a user), however, we rely on the inherent support eMoflon offers, which is discussed in Sect.~\ref{sec:implementation}.
Deriving repair rules to propagate such changes is possible in principle but remains future work. 

In practical applications, TGGs are often not only attributed but also equipped with \emph{attribute constraints}. 
These enable the user to, for example, link the values of attributes of correlated nodes. 
eMoflon comes with facilities to detect violations of such constraints and offers support to repair such violations. 
In our implementation, we rely on these features of eMoflon to support attribute constraints but do not contribute additional support in our newly proposed synchronization algorithm. 

To summarize, while fully formalized for the case of plain TGG rules without attribution, our implementation already supports the synchronization of attributed TGGs with additional attribute constraints. 
As these additional features do not affect our construction of short-cut and repair rules, we do not consider them (yet) to improve the propagation of attribute changes (that may lead to violations of attribute constraints). 
Instead, we rely on the existing theory and facilities of eMoflon as introduced by Anjorin et al.~\cite{AVS12}. 
In contrast, while computing short-cut and repair rules of rules with NACs is straightforward, adapting our synchronization algorithm to that case is future work and no tool support is available yet.

\section{Implementation}
\label{sec:implementation}
Our implementation\footnote{Both, the implementation and the evaluation, 
 can be accessed via \url{https://github.com/Echtzeitsysteme/STTT-SC-Eval}.} of a 
model synchronizer  using (shortcut) repair rules is built on top of the existing EMF-based, general-purpose graph and model transformation tool eMoflon~\cite{EL14,WAFVSL19,WARV19}.
eMoflon offers support for rule-based unidirectional and bidirectional graph transformations where the latter one uses TGGs.
The model synchronizer implemented in eMoflon extends Algorithm~\ref{alg:sync} slightly. It allows any kind of user edit on the source part of a triple graph. If there are no forward or repair rules to fix a broken match, broken rule applications can be revoked.
Revoking of rule applications has been the standard way of fixing broken matches. Hence, the implemented model synchronizer is a true extension of the previous synchronizer in eMoflon supporting the repair of broken applications. 

In the following, we present the architecture behind our optimized model synchronizer first.
Thereafter, we  describe how the automatic calculation of short-cut and repair rules is implemented. 

\begin{figure*}
	\centering
	\includegraphics[width=.9\textwidth]{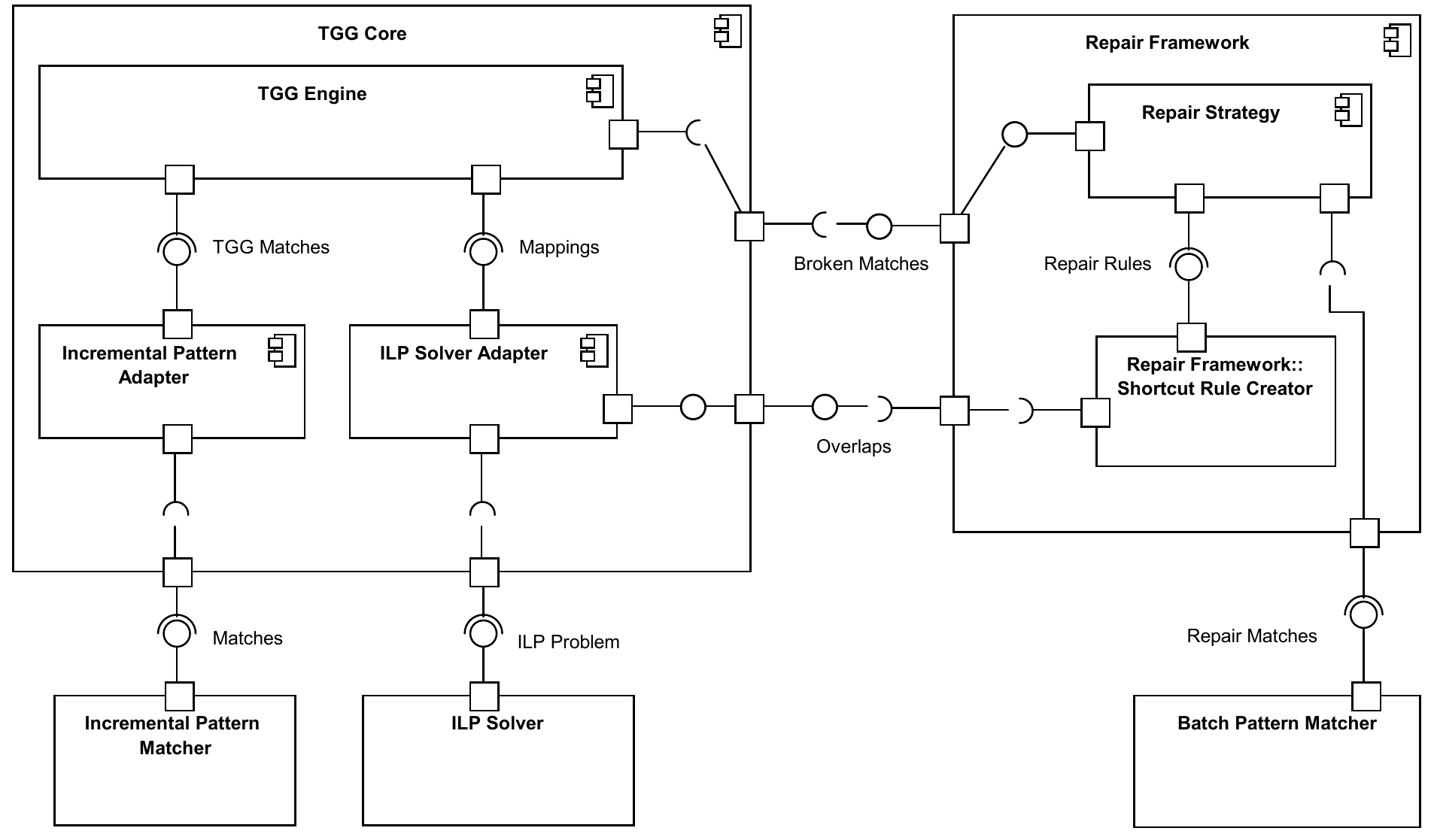}
	\caption{eMoflon -- Architecture of the bidirectional transformation engine}
	\label{fig:architecture}
\end{figure*}

\subsection{Tool architecture}
Figure~\ref{fig:architecture} depicts a UML component diagram to show the main components of eMoflon's bidirectional transformation engine.
The architecture has two main components: 
 \textit{TGG Core} contains the core components of eMoflon and \textit{Repair Framework} adds (short-cut) repair rules to eMoflon\rq{}s functionality. 
The \emph{TGG engine} manages the synchronization process and alters source, target, and correspondence model in order to restore consistency. 
For this purpose, it applies for\-ward/ back\-ward \textit{operationalized TGG rules} to translate elements or revokes broken rule applications. 

Finding matches in an incremental way is an important requirement for efficient model synchronization since minor model changes should be detectable without re-evaluating the whole model.
For this reason, eMoflon relies on \textit{incremental pattern matching} to detect the appearance of new matches as well as the disappearance of formerly detected ones. 
It uses different incremental pattern matchers such as Democles~\cite{VD13} and HiPE~\cite{hipe} and allows to switch freely between them for optimizing the performance for each transformation scenario.
Furthermore, eMoflon employs the use of various \textit{integer linear programming} (ILP) solvers such as Gurobi~\cite{Gurobi16} and CPLEX~\cite{CPLEX}, e.g., in order to find correspondence links (mappings) between source and target models, which is referred to as consistency check~\cite{leblebici2017inter}.

We have extended this basic setup by introducing the \emph{Repair Framework}, which consists of the \emph{Repair Strategy} and the \emph{Shortcut Rule Creator}. 
The \emph{Repair Strategy} is attached to the \textit{TGG Engine} from which it is called with a set of broken rule matches.
It attempts to repair the corresponding rule applications by using repair rules created by the \emph{Shortcut Rule Creator}, which uses the ILP interface provided by the \emph{TGG Core} in order to find overlaps between TGG rules and finally, to create short-cut repair rules.
For invoking the repair rules, however, we have to find matches of repair rules.
This is done by a \textit{Batch (local-search) Pattern Matcher} which, in contrast to the incremental pattern matcher, does not perform any book-keeping. 
As a repair of a rule application is always done locally, the checking of matches throughout the whole model is considered to be too expensive and thus, a \emph{Batch Pattern Matcher} can perform this task more efficiently.

\subsection{ILP-based short-cut rule creation}
\label{sec:sc-rule-creation} 
In order to create an overlap between two rules, a morphism between the graphs of both rules has to be found: 
Each element may only be mapped once; a context element may only be mapped to another context element. Created elements are mapped to each other, respectively. 
Furthermore, a node can only be mapped to a node of the same type  as we do not  incorporate inheritance between types yet.
Edges are allowed to be mapped to each other only if their corresponding source and target elements are also mapped to each other, respectively.

We use integer linear programming (ILP) to encode the search space of all possible mappings and search for a maximal mapping.
Each possible mapping $m$ is considered to be a variable of our ILP problem such that calculating 
\begin{equation*}
	max (\sum_{m \in M} m) 
\end{equation*}
yields the maximal overlap, with $ M $ being the set of all mappings and $ m \in \{0, 1\} $.
To ensure that each element $e$ is mapped only once, we define a constraint to exclude non-used mappings: $ (\sum_{m \in A_e} m) \leqslant 1 $ with $ A_e $ being the set of all alternative mappings for element $e$.
To ensure that edges are mapped only if their adjacent nodes are mapped as well, we define the following constraint: $ m_e \implies m_v $ which translates to $ m_e \leq m_v $ with $ m_e $ being the edge mapping and $ m_v $ being one of the mappings of node $src(e)$ or $trg(e)$. 
Maximizing the number of activated variables yields the common kernel of both input rules, i.e., a maximal overlap between them.
If the overlap between the created elements of both rules is empty, we drop this overlap as the resulting short-cut rule would not preserve any elements. 
Given a common kernel of two rules, we glue them along this kernel and yield a short-cut rule. 
For all elements of the resulting short-cut rule, which are not in the common kernel, we do the following: (1) Preserved elements remain preserved in the short-cut rule. (2) Created elements of the first rule become deleted ones as the first rule is inverted. (3) Created elements of the second rule remain created ones. 

We calculate two kinds of overlap for each pair of rules and hence, two short-cut rules: a maximal and a minimal overlap. 
The maximal overlap is calculated by allowing mappings between all created and context elements, respectively.
On the other hand, the minimal overlap is created by allowing mappings between created elements only.
Considering the corresponding ILP problem, this means that all other mapping candidates are dropped. 

Finally, the derived short-cut rules are operationalized to obtain the repair rules employed in our synchronization algorithm. 

\subsection{Attribute Constraints}
Although attribute constraints have not been incorporated formally in our approach, eMoflon is able to define and solve those within the former legacy translation and synchronization process.
As can be seen in Fig.~\ref{fig:eval_rule_set_full}, many rules have an equality constraint defined between the name attributes of created elements on both, source and target parts.
For TGG rules, this means that the attribute values may be chosen arbitrarily since both nodes would be created from scratch.
In forward rules, source elements are already present which means that an attribute constraint can be interpreted as to propagate or copy the already present value to a newly created element. 
We reuse this functionality for our new synchronization process in the following a way:
After applying a repair rule, we ensure that the constraints of the replacing rule are fulfilled. 
The definition of attribute constraints and their treatment is due to Anjorin et al.~\cite{AVS12}.\footnote{This approach allows to specify constraints on attributes that involve also operations which are not only equality checks such as the concatenation of values of type \emph{String}.}

\section{Evaluation}
\label{sec:implAndEvaluation}
We evaluate our approach with respect to two aspects using the running example in an extended form. 
First, we investigate the performance of our approach w.r.t. information loss and execution time. 
A set of real and synthesized models is given which we use to apply four different kinds of model changes. 
Secondly, we evaluate the quality of our short-cut rule generation strategy by comparing generated short-cut rules with well-known code refactorings. 

Our experimental setup consists of 24 TGG rules (shown in Sect.~\ref{sec:evaluation-ruleset}) that specify consistency between Java AST and custom documentation models. 
In addition, there are 38 short-cut rules being derived from the set of TGG rules.
A small modified excerpt of this rule set was given in Sect.~\ref{sec:example}. 
For this evaluation, however, we define consistency not only between \package{} and \folder{} hierarchies but also between type definitions, e.g., \classes{} and \interfaces{}, and  \fields{} and \methods{} with their corresponding documentation entries.

\subsection{Performance Evaluation}
To get realistic models, we extracted five models from Java projects hosted on Github using the reverse engineering tool MoDisco~\cite{HB14} and translated them into our own documentation structure. 
In addition, we generated five synthetic models consisting of $n$-level \package{} hierarchies with each non-leaf \package{} containing five sub-\packages{} and each leaf \package{} containing five \classes{}. 
While the realistic models shall show that our approach scales to real world cases, the synthetic models are chosen to show scalability in a more controlled way by increasing hierarchies gradually. 

To evaluate our synchronization process, we performed several model changes.
We refactored each of the models in four different scenarios; two example refactorings are the moving of a \class{} from one \package{} to another or the complete relocation of a \package{}.
Then we used eMoflon to synchronize these changes in order to restore consistency to the documentation model using two synchronization processes, namely with and without \repairRules{}. 
The legacy synchronization process of eMoflon is presented in~\cite{LAFVS17,Leblebici18}; the new synchronization process applying additional repair rules takes place according to the algorithm presented in Sect.~\ref{sec:syncProcess} with the extensions mentioned in Sect.~\ref{sec:prospect}.

These synchronization steps are subject to our evaluation and we pose the following research questions:
\textbf{(RQ1)} {\em For different kinds of model changes, how many elements can be preserved that would be deleted and recreated otherwise?}
\textbf{(RQ2)} {\em How does our new synchronization process affect the runtime performance?}
\textbf{(RQ3)} {\em Are there specific scenarios in which our new synchronization process performs especially good or bad?}

In the following, we evaluate our new {\em synchronization process by repair rules} against the {\em legacy synchronization process} in eMoflon. While the legacy one revokes forward rule applications and re-propagates the source model using forward rules, our new one prefers to apply short-cut repair rules as far as possible and falls back to revoking and re-propagation if there is no possible repair rule application. 

To evaluate the performance of the legacy and the new model synchronization processes, we consider the following synchronization scenarios:
Altering a root \package{} by creating a new \package{} as root would imply that many rule applications have to be reverted to synchronize the changes correctly with the legacy synchronization process (Scenario 1).
In contrast, our new approach might perform poorly when a model change does not inflict a large cascade of invalid rule applications.
Hence, we move \classes{} between \packages{} (Scenario 3) and \methods{} between \classes{} (Scenario 4) to measure if the effort of applying \repairRules{} does infer a performance loss when both, the new and old algorithm, do not have to repair many broken rule applications.
Note that Scenario 4 extends our evaluation presented in \cite{FKST19} as it provides a more fine-granular scenario.
Finally, we simulate a scenario which is somewhat between the first three by relocating leaf \packages{} (Scenario 2) which, using the legacy model synchronization, would lead to a re-translation of all underlying elements. 

\begin{table*}
	\centering
	\caption{Legacy synchronizer -- Time in sec. and number of created elements}
	\label{tbl:measurements_legacy}
	\begin{tabular}{@{}l@{\hskip 7.5pt}*{15}{l}@{}}
		\toprule
		& \multicolumn{2}{c}{Both}	& \phantom{a}	& \multicolumn{11}{c}{Legacy Synchronization}	\\
		\cmidrule{2-3} \cmidrule{5-15}
		& \multicolumn{2}{c}{Trans.}	& \phantom{i}	& \multicolumn{2}{c}{Scen. 1}	& \phantom{i}	& \multicolumn{2}{c}{Scen. 2}	& \phantom{i}	& \multicolumn{2}{c}{Scen. 3} & \phantom{i}	& \multicolumn{2}{c}{Scen. 4}\\
							\cmidrule{2-3} 		\cmidrule{5-6} 		\cmidrule{8-9} 		\cmidrule{11-12}	\cmidrule{14-15}
		Models				& Sec	& Elts	&	& Sec	& Elts	&	& Sec	& Elts	&	& Sec	& Elts 	& 	& Sec	& Elts	\\
		\midrule
		lang.List			& 0.3 	& 25 	&	& 0.2	& 20	&	& --	& --	&	& 0.06	& 5	 	& 	& 0.04		&	3 	\\
		tgg.core			& 6.4	& 1.6k	&	& 39	& 1.6k	&	& 3.8	& 99	&	& 0.64	& 17 	& 	& 0.2		&	3	\\
		modisco.java		& 9.9	& 3.2k	&	& 228	& 3.3k	&	& 18.6	& 192	&	& 3.6	& 33 	& 	& 0.4		&	4	\\
		eclipse.compare		& 10.74	& 3.8k	&	& 83	& 3.7k	&	& 3.1	& 76	&	& 2.36	& 47 	& 	& 0.1		&	1	\\
		eclipse.graphiti	& 20.7	& 6.5k	&	& 704	& 6.5k	&	& 63.9	& 490	&	& 5.65	& 25 	& 	& 0.9		&	3	\\
		\midrule
		synthetic $n=1$		& 0.6	& 89	&	& 0.5	& 84	&	& 0.2	& 21	&	& 0.07	& 5	 	& 	& 0.03		&	1	\\
		synthetic $n=2$		& 1.4	& 345	&	& 1.7	& 340	&	& 0.2	& 21	&	& 0.11	& 5		& 	& 0.04		&	1   \\
		synthetic $n=3$		& 3.5	& 1369	&	& 13.2	& 1364	&	& 0.3	& 21	&	& 0.11	& 5		& 	& 0.07		&	1	\\
		synthetic $n=4$		& 14.5  & 5.5k	&	& 141.5	& 5.5k	&	& 1		& 21	&	& 0.32	& 5 	& 	& 0.09		&	1	\\
		synthetic $n=5$		& 58.5	& 22k	&	& 2863	& 22k	&	& 10.7	& 21	&	& 1.07	& 5 	& 	& 0.23		&	1	\\
		\bottomrule
	\end{tabular}
\end{table*}

\begin{table*}
	\centering
	\caption{New synchronizer -- Time in sec. and number of created elements}
	\label{tbl:measurements_sc}
	\begin{tabular}{@{}lllllllllllll@{}}
		\toprule
		& \multicolumn{11}{c}{Synchronization by Repair Rules}	\\
		\cmidrule{2-12}
		& \multicolumn{2}{c}{Scen. 1}	& \phantom{i}	& \multicolumn{2}{c}{Scen. 2}	& \phantom{i}	& \multicolumn{2}{c}{Scen. 3} & \phantom{i}	& \multicolumn{2}{c}{Scen. 4}\\
		\cmidrule{2-3} 		\cmidrule{5-6} 		\cmidrule{8-9} 		\cmidrule{11-12}	
		Models				& Sec	& Elts	&	& Sec	& Elts	&	& Sec	& Elts 	& 	& Sec	& Elts	\\
		\midrule
		lang.List			& 0.2	& 0		&	& --	& --	&	& 0.03	& 0 	& 	& 0.02		&	0	\\
		tgg.core			& 0.8	& 0		&	& 0.11	& 0		&	& 0.05	& 0 	& 	& 0.04		&	0	\\
		modisco.java		& 2.5	& 0		&	& 0.2	& 0		&	& 0.09	& 0 	& 	& 0.1		&	0	\\
		eclipse.compare		& 0.7	& 0		&	& 0.08	& 0		&	& 0.04	& 0 	& 	& 0.03		&	0	\\
		eclipse.graphiti	& 6.1	& 0		&	& 0.21	& 0		&	& 0.09	& 0 	& 	& 0.1		&	0	\\
		\midrule
		synthetic $n=1$		& 0.1	& 0		&	& 0.05	& 0		&	& 0.03	& 0 	& 	& 0.05	&	0	\\
		synthetic $n=2$		& 0.1	& 0		&	& 0.05	& 0		&	& 0.02	& 0 	& 	& 0.04		&	0	\\
		synthetic $n=3$		& 0.1	& 0		&	& 0.07	& 0		&	& 0.02	& 0 	& 	& 0.04		&	0	\\
		synthetic $n=4$		& 0.4	& 0		&	& 0.14	& 0		&	& 0.04	& 0 	& 	& 0.04		&	0	\\
		synthetic $n=5$		& 1.5	& 0		&	& 0.37	& 0		&	& 0.09	& 0 	& 	& 0.06		&	0	\\
		\bottomrule
	\end{tabular}
\end{table*}

Tables~\ref{tbl:measurements_legacy} and \ref{tbl:measurements_sc} depict the measured time in seconds (Sec) and the number of re-/created elements (Elts) in each scenario (1)--(4). 
The first table additionally shows measurements for the initial translation (Trans.) of the Java AST model into the documentation structure.  
For each scenario, Table~\ref{tbl:measurements_legacy} shows the numbers of synchronization steps using the legacy synchronizer without \repairRules{} while Table~\ref{tbl:measurements_sc} reflects the numbers of our new synchronizer with \repairRules{}.

W.r.t. our research questions stated above, we interpret these tables as follows:
The Elts columns of Table~\ref{tbl:measurements_sc} show clearly that using repair rules preserves all those elements in our scenarios that are deleted and recreated by the legacy algorithm otherwise as shown in Table~\ref{tbl:measurements_legacy} \textbf{(RQ1)}.
The runtime shows a significant performance gain for Scenario~1 including a worst-case model change in which the legacy algorithm has to re-translate all elements \textbf{(RQ2)}. 

{\em Repair rules} do not introduce an overhead compared to the legacy algorithm as can be seen for the synthetic time measurements in Scenario 4 where only one rule application has to be repaired or reapplied \textbf{(RQ2)}.
Our new approach excels when the cascade of invalidated rule applications is long. Even if this is not the case, it does not introduce any measurable overhead compared to the legacy algorithm as shown in Scenarios 2, 3, and 4 (\textbf{RQ3}). 

\paragraph{Threats to validity.}
Our evaluation is based on five real world and five synthetic models. 
Of course, there exists a wide range of Java projects that differ significantly from each other w.r.t. their size, purpose, and developer style.
Thus, the results may not be transferable to other projects.
Nonetheless, we argue that the four larger models extracted from Github projects are representative since they are deduced from established tools of the Eclipse ecosystem.
The synthetic models are also representative as they show the scalability of our approach in a more controlled environment with an increasing scaling factor.
Together, realistic and synthetic models show that our approach does not only increase the performance of eMoflons synchronization process but also reduce the amount of re-created elements.
Since each re-created element may contain information that would be lost during the process, we preserve this information and increase the overall quality of eMoflons synchronization results. 
In this evaluation, we selected four edit operations that are representative w.r.t. their dependency on other edit operations. 
They may not be representative w.r.t. other aspects such as size or kind of change. We consider those aspects to be of minor importance in this context as dependency is the cause for deleting and recreating elements in the legacy synchronization process.
Finally, we limited our evaluation to one TGG rule set only as we experienced similar results for a broader range of TGGs from the eMoflon test zoo\footnote{Accessible via \url{https://github.com/eMoflon/emoflon-ibex-tests}}. 

\subsection{Refactorings}
\label{sec:refactorings}
As explained in Sect.~\ref{sec:implementation}, we currently employ two different strategies to overlap two rules and to create a short-cut rule.
We pose the following research question: \textbf{(RQ4)} {\em Are the generated \shortcutRules{} applicable to realistic scenarios? Are further \shortcutRules{} necessary?} 
Since our example addresses code changes that are incorporated by the Java AST model primarily, we relate our approach to available code refactorings.
In the following, we refer to the book on code refactorings written by Martin Fowler~\cite{Fowler2018} which presents 66 refactorings. 

Our example TGG, depicted in Fig.~\ref{fig:eval_rule_set_full}, defines consistency on a structural level solely, without incorporating behaviour, i.e., the bodies of methods and constructors.
Hence, we selected those refactorings that describe changes on \packages{}, \classes{} and \interfaces{}, {\em MethodDeclarations} and \parameters{}, and \fields{}.
The result is a set of 16 refactorings for which we evaluated if \shortcutRules{} help to directly propagate the corresponding change of the AST model or deletion and recreation has to take place.

\begin{figure*}
	\centering
	\includegraphics[width=.9\textwidth]{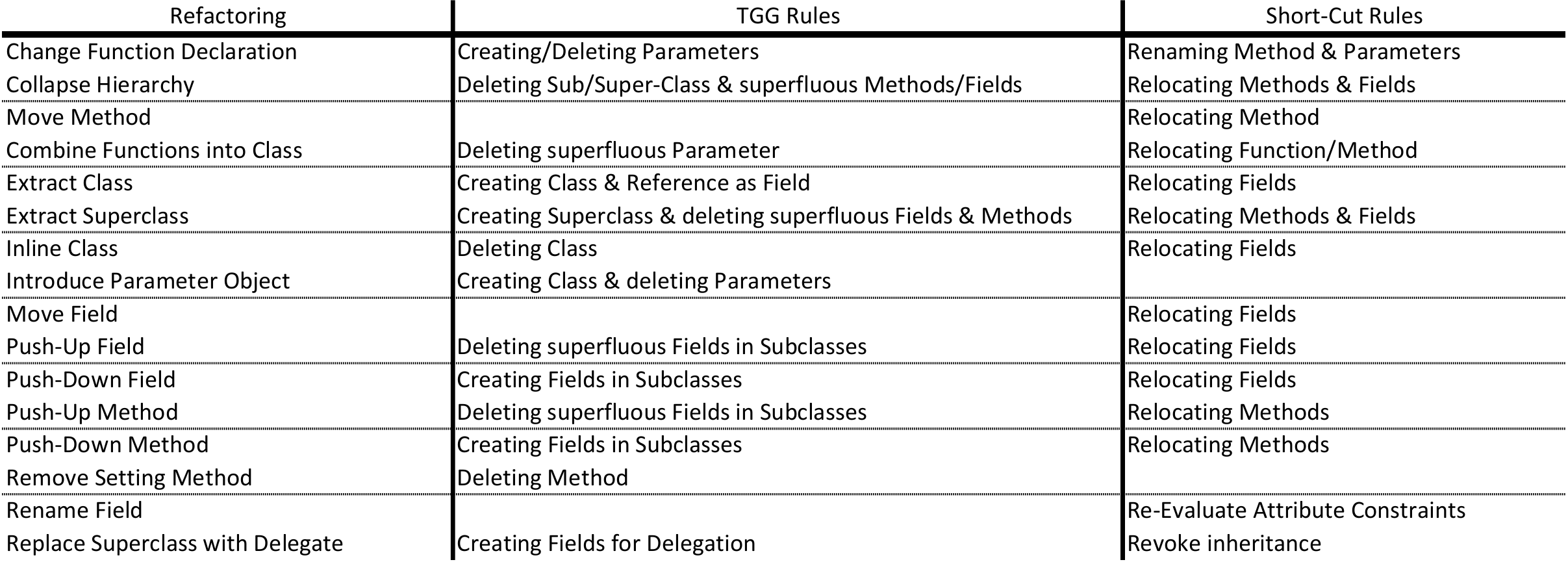}
	\caption{Refactorings}
	\label{fig:refactoring}
\end{figure*}

Fig.~\ref{fig:refactoring} lists these refactorings together with information on the TGG rules and/or \shortcutRules{} that are applicable in these scenarios. 
For some of the refactorings as e.g., {\em Extract Class} and \emph{Push-Down Field}, we identified situations where not only \shortcutRules{} are necessary to propagate the changes.
In these cases, new elements may be created which can be propagated using operationalized TGG rules.
The deletion of elements can be propagated by revoking the corresponding prior propagation step.
However, many refactorings benefit from using \shortcutRules{}, for example, those that move methods and fields. 
If recreation of documentation on the target part is necessary, it can lead to information loss as there may not be all the necessary information  in the Java AST model.

\textit{Example:} \emph{Push-Up Field} moves and merges a similar field from various subclasses into a common superclass.
If one of the subclass fields is moved to the superclass, we can propagate this change using \emph{Move-Field-Repair-Rule}, which is depicted in Fig.~\ref{fig:move_field_rule}.

\begin{figure}
	\includegraphics[width=\columnwidth]{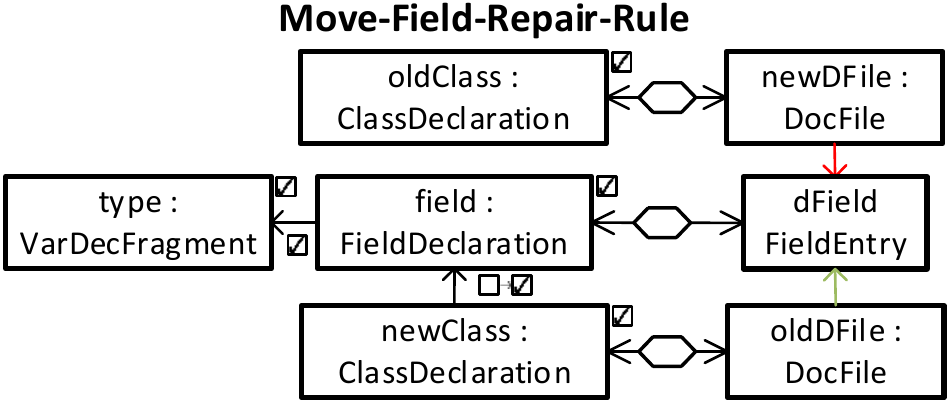}
	\caption{Move-Field-Repair-Rule}
	\label{fig:move_field_rule}
\end{figure}

In summary, we are able to solve all 16 refactorings using a combination of (inverse) TGG rules and our generated \shortcutRules{} \textbf{(RQ4)}. 

\paragraph{Threats to validity.}
Note that \shortcutRules{} are especially useful when elements are moved instead of deleting and recreating them in some other location.
Those changes are hard to detect and are not covered here.
Refactorings such as {\em Push-Up Method}, which moves a method that occurs in several subclasses to their common superclass, can be done in two different ways.
First, one of the methods is moved to the superclass while the methods in the other subclasses are deleted.
This employs the use of \shortcutRules{} for the moved method followed by revocation steps for the deleted methods to delete the corresponding documentation elements.
Second, all methods may be deleted and a new similar method is created in the superclass.
In that case, there is no \shortcutRule{} that helps to preserve information and all propagated documentation elements for the method will be blank.
Hence, our approach depends on the kind of change. In particular, it helps when user edits also try to preserve information instead of recreating them.

In addition, we have not incorporated behaviour in our example; such an extension of our TGG may be considered in future work.  
However, we can argue that most of those refactorings can be reduced to the movement of elements, the deletion of superfluous elements and the creation of new elements.
These changes are manageable in general using a sequence of \shortcutRule{} and (inverse) operationalized TGG rule applications.

Finally, we evaluated these cases by hand based on the generated \shortcutRules{} from our implementation.
Test cases implementing the identified refactorings and combinations of them will be accessible via eMoflons test zoo.

\section{Related Work}
\label{sec:related-work}
In this section, we relate our new model synchronization approach to already existing incremental model synchronization approaches. 
First, we discuss other TGG-based approaches in detail before relating to other bidirectional transformation (bx) approaches; these are considered more roughly. 
Finally, we mention some unidirectional approaches that are closely related to incremental model transformation and model repair. 
Work that is related to our use of partial triple graphs but not to model synchronization is considered in~\cite{KFST19}.

\paragraph{TGG-based approaches to incremental model synchronization.}
Synchronization approaches are supposed to comply with the least-change property, which means that no unnecessary deletions and thus information loss should take place while restoring consistency.
An overview of TGG-based least-change synchronization has been given by Stojkovic et al.~\cite{SLA17}. 
The first part of our related work is based on that presentation. 

Several approaches to model synchronization based on TGGs suffer from the fact that the revocation of a rule application may trigger the revocation of all dependent rule applications as well~\cite{GW09,LAVS12,LAFVS17,Leblebici18}. 
Such cascades of deletions shall be avoided to decrease runtime and unnecessary information loss. 

Leveraging an incremental pattern matcher for TGG-based model synchronization was first suggested in~\cite{LAFVS17,Leblebici18}. 
Proofs of termination, correctness, and completeness are given. 
Moreover, the approach is implemented. 
In fact, this is the legacy synchronization we evaluated against in Section~\ref{sec:implAndEvaluation}. 
As already mentioned, that approach revokes invalid consistency matches as long as there are any and subsequently, applies forward rules to translate yet untranslated elements. 
So, that approach is a typical example where a lot of unnecessary deletions may take place.

Hermann et al.~\cite{HEOCDXGE15} proposed a synchronization algorithm where, after an edit on the source part, first those correspondence elements are deleted that do not refer to an element in the source graph any longer. 
Thereafter, they parse the remaining triple graph to find the maximal, still valid sub-model. 
This model is used as a starting point to propagate the remaining changes from source to correspondence and target graphs using forward rules. 
The approach is completely formalized and proven to be correct, also for attributed TGGs; it can be applied to TGGs with deterministic\footnote{Deterministic in the sense that there are no competing rules for any translated element.} sets of operationalized rules. 
That approach  avoids some unnecessary deletions but there are some that still can occur. 
In fact, the amount of unnecessary deletion taking place in that approach is dependent on the given TGG rules; a concrete example for that is given in~\cite{SLA17}. 
While that approach is definitely a valuable contribution towards least-change synchronization, repeated parsing for maximally consistent sub-models is highly inefficient and might not scale to large models. 
At least part of that approach is implemented as \textsc{HenshinTGG}~\cite{EHGB12} using \textsc{AGG}~\cite{Taentzer03} to perform necessary dependency checks on derived rules. 
As that approach focusses on correctness, completeness, and invertibility, the amount of achieved incrementality as well as principles of least change are not discussed in~\cite{HEOCDXGE15}.  

In~\cite{GH09}, Giese and Hildebrandt propose rules that save nodes instead of deleting and re-creating them. 
In particular, they present a rule that directly propagates the movement of elements, i.e., the redirection of edges between existing elements. 
Moreover, they suggest to try a re-use of elements before deleting them. 
But they neither present a general construction for their rules nor formalize the re-use that takes place. 
Consequently, no proof of correctness is given. 
Instead, it is left as future work in~\cite{GHL14}. 
The additional propagation rules that are given exemplary in~\cite{GH09} can be automatically derived as \repairRules{} using our approach. 
In~\cite{BPDSD14}, Blouin et al. also add specifically designed repair rules to the rule set of their case study for avoiding information loss. 
Those example rules can be realized as \repairRule{} in our approach as well. 

In a similar vein, Greenyer et al.~\cite{GPR11} propose to delete elements not directly but to mark them for deletion and to allow for their re-use in rule applications during synchronization. 
Only elements that cannot be re-used are deleted at the very end of synchronization. 
But that approach comes without any formalization and proof of correctness as well. 

In contrast, the idea of re-using elements in model synchronizations has been rigorously formalized by Orejas and Pino~\cite{OP14}. 
They introduced \emph{forward translation rules with reuse} and proposed a synchronization algorithm based on those rules. 
That algorithm is actually proven to be correct; moreover, it is incremental (in a technical sense). 
The practical effects of applying a \repairRule{} in our approach and in their approach are very similar. 
While our \repairRules{} allow for reuse and perform necessary deletions on the correspondence and target parts directly, their forward translation rules allow for a reuse where necessary deletions are performed at the end of a synchronization in a separate step. 
They need some additional technical infrastructure to determine the exact amount of necessary deletion. 
To the best of our knowledge, their approach has not been implemented yet.  

In a guideline on how to develop a TGG, Anjorin et al.~\cite{ALKSS15} explain how certain kinds of rules in a TGG avoid the loss of information better than others. 
There is empirical evidence that, following these guidelines, synchronization can be considerably accelerated compared to a batch mode as long as there is no need for additional offline recognition of model differences~\cite{LAS14}.
Transforming a given TGG into that form, however, may change the defined language and thus, is not always applicable. 
For example, the grammar of our running example allows to generate hierarchies of \packages{} that constitute a set of disconnected trees. 
For meeting the suggestions in~\cite{ALKSS15}, a naive change of this grammar may change the language such that arbitrary graphs can be generated. 
That effect can be avoided by, e.g., designing suitable NACs for the rules and proving the equality of the generated model languages. 
That effort is not needed when following our approach. 

\begin{table*}
	\centering
	\caption{An overview of TGG-based synchronization approaches}
	\label{tbl:overview-relatedworks}
	\begin{tabular}{@{}lllllll@{}}
		\toprule
		 & Degree of information loss	& Automated	& Correctness proven	& Implemented	& Evaluated performance gain \\
		\midrule
		\cite{LAFVS17,Leblebici18} & high & yes & yes & yes & yes \\
		\cite{GW09}	& high	& yes	& only partially in \cite{GHL14}	& yes	& yes\\
		\cite{LAVS12} & high & yes & yes & yes & yes\\
		\cite{HEOCDXGE15} & to some extent & yes & yes & at least partially & no \\ 
		\cite{GH09} & low & yes & no & yes & yes \\
		\cite{GPR11} & low & yes & no & yes & no  \\
		\cite{OP14} & low & yes & yes & no & no \\
		\cite{ALKSS15,LAS14} & low & not needed & not needed & yes & yes \\
		\cite{TOLR17,OPKK18} & low & no & yes & yes & no \\
		ours	& low & yes & yes & yes & yes \\
		\bottomrule
	\end{tabular}
\end{table*} 

In summary, it is well-known in the literature that there are a lot of situations where the derived forward rules of a TGG (and the revocation of their applications) are not suitable to efficiently propagate changes from source to target models. 
Several formal and informal approaches have been suggested to avoid this problem, at least partly. 
Table~\ref{tbl:overview-relatedworks} provides an overview of all the approaches  described above. 
It indicates the degree of information loss and presents whether the approach is automated, whether correctness of the proposed synchronization algorithm is proven, whether it has been  (prototypically) implemented, and whether any performance gain could be shown for it. 
Our approach is based on the automated derivation of \repairRules{}; it is able to comply with all the above categories. 
The correctness has been shown for model synchronization with repair rules. 
As our implemented synchronization process can also revoke forward rules, the correctness proof has to be slightly extended to cover also that case which seems to be straight forward (see discussion in Sect.~\ref{sec:prospect}).
Furthermore, support for some additional features of TGGs like NACs and attribution is future work (NACs) or not rigorously formalized (attribution).

\paragraph{Comparison to other bx approaches.} 
Anjorin et al.~\cite{AA17} compared three state-of-the-art bx tools, namely eMoflon~\cite{EL14} (rule-based), mediniQVT~\cite{mediniQVT} (constraint-based), and BiGUL \cite{KoZH16} (bx programming language) w.r.t. model synchronization.
They point out that synchronization with eMoflon is faster than with both other tools as the runtimes of those tools all  correlate with the overall model size while the runtime of eMoflon correlates with the size of the changes done by edit operations.
Furthermore, eMoflon is the only tool that was able to solve all but one synchronization scenario while mediniQVT failed in four and BiGUL in two scenarios. 
One scenario was not solved because the solution with eMoflon deletes more model elements than absolutely necessary in that case.
Using short-cut repair rules, we can solve the remaining scenario and moreover, can further increase the performance of eMoflon when solving model synchronization tasks. 
Macedo and Cunha present bidirectional model transformations based on ATL in~\cite{MC16}. 
By using the SAT solver Alloy, they are able to guarantee least-change model synchronization where two metrics  are supported measuring change: the \emph{graph edit distance} and the \emph{operation-based distance}. 
While the synchronization results may be very good, this solver-based approach does not scale for large models. 
All this suggests that our tool is highly competitive, not only among TGG-based tools but also in comparison to other bx tools. 

With regard to theoretical considerations, least change and incremental synchronization have also been actively investigated in other approaches, in particular when using lenses, e.g., \cite{DXC11,WGW11,HPW12,HPC18,HB19}. 
The approach by Wang et al. \cite{WGW11} seems to be the most similar one to ours. 
That approach derives functions to directly propagate changes from a source to a view and is applicable to tree-shaped data structures. 
As those approaches are less close to our work, detailed formal comparisons are left to future work. 

\paragraph{Further related works.} 
\emph{Change-preserving model repair} as presented in~\cite{TOLR17,OPKK18} is closely related to our approach.
Assuming a set of consistency-preserving rules and a set of edit rules to be given, each edit rule is accompanied by one or more repair rules completing the edit step if possible. 
Such a complement rule is considered as repair rule of an edit rule w.r.t. an overarching consistency-preserving rule. 
Operationalized TGG rules fit into that approach but provide more structure: As graphs and rules are structured in triples, a source rule is also an edit rule being complemented by a forward rule. 
In contrast to that approach, source and forward rules can be automatically deduced from a given TGG rule. 
By our use of short-cut rules, we introduce a pre-processing step to first enlarge the sets of consistency-preserving rules and edit rules. 
Furthermore, the repair process presented in that paper has more restrictive presumptions than our synchronization process using repair rules w.r.t. independence of rule applications.

Boronat~\cite{Boronat19} presents an incremental uni-directional transformation approach.  
When retranslating a model after a change, affected elements of the old model are marked first and then, if possible, re-used instead of deleted and re-created (similar to the approaches suggested in~\cite{GPR11,OP14} for TGGs). 
Again, the same effects can be obtained by constructing and applying short-cut rules but there, for plain graph transformation. 
A correctness proof for that approach is still missing.

\section{Conclusion}
\label{sec:conclusion}
Model synchronization, i.e., the task of restoring the consistency between two models after model changes, poses challenges to modern bidirectional model transformation ap\-proach\-es and tools:
We expect them to synchronize changes without unnecessary loss of information and to show a reasonable performance.
Here, we restrict ourselves to model synchronizations where only one model is changed at a time.

While Triple Graph Grammars (TGGs) provide the means to perform model synchronization tasks in general, efficient model synchronization without unnecessary information loss may not always be fulfilled since basic TGG rules are not designed 
to support intermediate model editing and repair.
Therefore, we propose to add \shortcutRules{}, a special form of generalized TGG rules that allow to take back one edit action and to perform an alternative one. In our evaluation, we show that repair rules derived from \shortcutRules{} allow for a kind of incremental model synchronization with considerably decreased information loss and improved runtime compared to synchronization without these rules.  

In this paper, we show the correctness of our synchronization approach, present the implementation design, and evaluate the corresponding tool support w.r.t. performance and unnecessary information loss. 
While the tool support already covers attributes of model elements, the correctness proof of our synchronization approach w.r.t. to these extensions is prepared but still up to future work. 

While model synchronization means the propagation of model changes from one view to another, model changes may also occur concurrently on both views of a model. 
Hence, model synchronization approaches have to cover those scenarios as well. 
Short-cut rules may also be promising to avoid information loss in that more general setting; they have not been considered in the context of other approaches to concurrent model synchronization in the literature~\cite{OBE+13,XSHT13}. 
As changes of both model views may be in conflict with each other, the development of an efficient concurrent model synchronization process which avoids unnecessary information loss poses a challenge for future work.

\subsubsection*{Acknowledgments}
This work was partially funded by the German Research Foundation (DFG) project \enquote{Triple Graph Grammars (TGG)~2.0}. \\
We would also like to thank the anonymous reviewers for their thoughtful comments and efforts.

\bibliographystyle{spmpsci}
\bibliography{literature}

\appendix

\section{Evaluation Ruleset}
\label{sec:evaluation-ruleset}
In this section, we present additional information related to our evaluation from Sect.~\ref{sec:implAndEvaluation}.

Figure~\ref{fig:eval_rule_set_full} depicts the full TGG rule set used of our evaluation.
The first rule \textit{JavaModel-2-DocModel-Rule} defines consistency between a MoDisco \model{} and a \docModel{} that contains three sub \docModels{} and another \folder{} linked to the common \docModel{}.
These different containers are used to separate Java entities on the documentation site to split them up into common Java data types, external Java references and source references.
\textit{JavaModel-2-DocModel-Rule} then defines consistency between \packages{} and \folders{} given that their parent are a MoDisco \model{} and a \docModel{}, respectively.
Using \textit{JavaPackage-2-DocFolder-Rule}, we can now create \package{} and \folder{} hierarchies recursively.
Furthermore, there are four rules that define consistency for \classDecs{}, \interfaceDecs{}, \enumDecs{} and inner \classDecs{} each with a \doc{}.
Also, for the nine primitive types, e.g., boolean, byte and short, consistency is defined between each of them and a \doc{}.
Given a \classDec{} or an \interfaceDec{} with its corresponding \doc{}, we also define consistency between \methodDecs{} on one and \methodEntries{} on the other side.
Using the consistency between methods on both sides, we are able to define consistency between \typeAccesses{} and \parameters{}, once for method signatures and once for the return statement.
Finally, we define consistency between generalization and realization relationships using three rules.
First, a rule for \classDecs{} that extend another \classDec{}, second a rule for \interfaceDecs{} extending another \interfaceDec{} and last for \classDecs{} implementing an \interfaceDec{}.
\begin{figure*}
	\centering
	\includegraphics[width=0.79\textwidth]{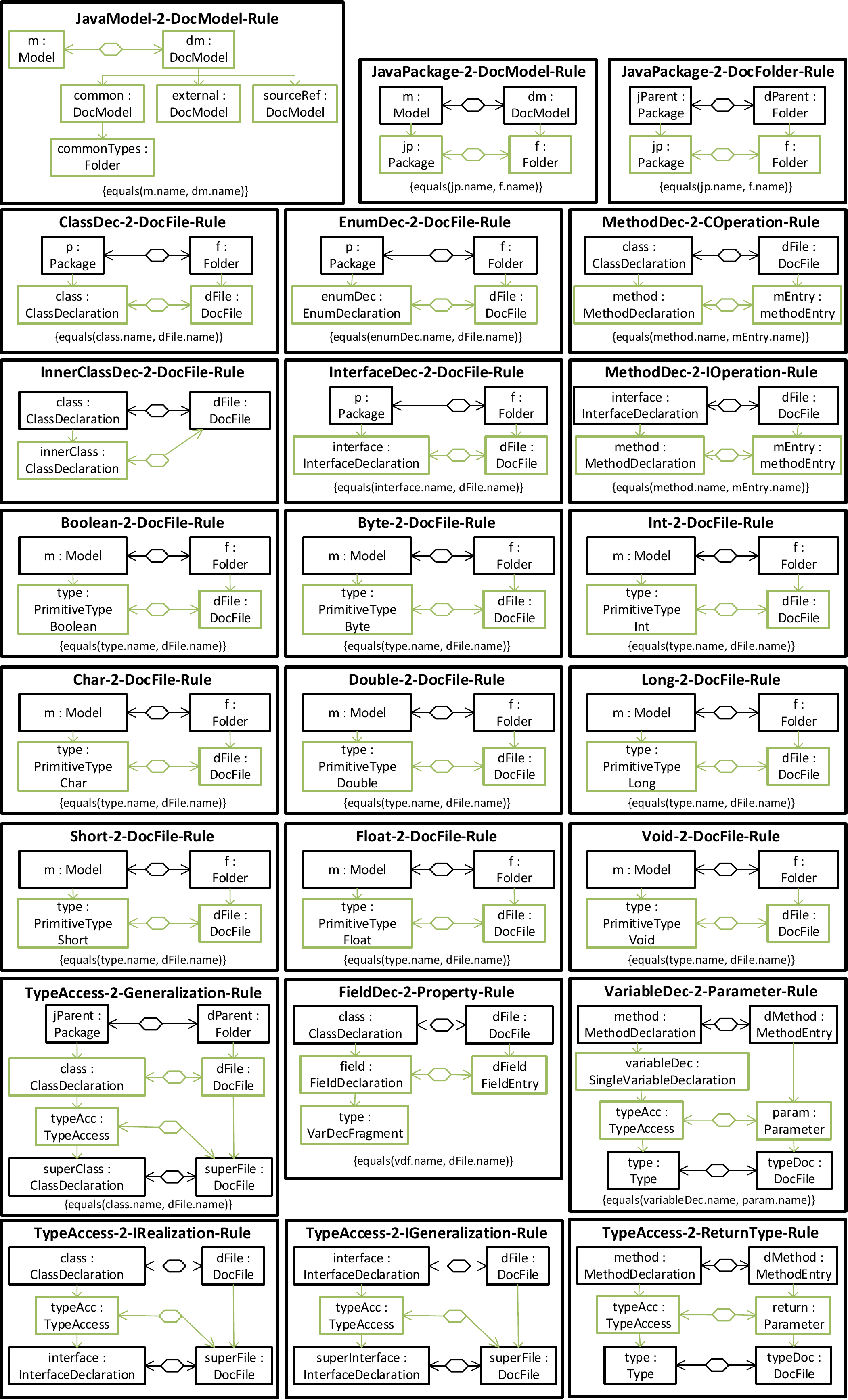}
	\caption{Evaluation -- TGG Rule Set}
	\label{fig:eval_rule_set_full}
\end{figure*}

\end{document}